\documentclass[journal,final,twocolumn,twoside]{IEEEtran}

\usepackage{amsmath,amsfonts}
\usepackage{algorithmic}
\usepackage{algorithm}
\usepackage{amssymb}
\usepackage{subfigure}

\allowdisplaybreaks[2]
\usepackage{float} 

\usepackage{flushend}
\usepackage{array}
\usepackage[caption=false,font=normalsize,labelfont=sf,textfont=sf]{subfig}
\usepackage{textcomp}
\usepackage{stfloats}
\usepackage{url}
\usepackage{verbatim}
\usepackage{multirow} 
\usepackage{ntheorem}
\usepackage{graphicx}
\usepackage{color}
\usepackage{cite}
\theorembodyfont{\upshape}
\theoremheaderfont{\it}
\qedsymbol{\square}
\theoremseparator{:}
\newtheorem{theorem}{\indent Theorem}
\newtheorem{lemma}{\indent Lemma}

\newtheorem*{proof}{\indent Proof}
\newtheorem{remark}{\indent Remark}

\makeatletter

\newcommand{\Rmnum}[1]{\expandafter\@slowromancap\romannumeral #1@}
\makeatother

\begin{document}

\makeatletter
\let\myorg@bibitem\bibitem
\def\bibitem#1#2\par{%
  \@ifundefined{bibitem@#1}{%
    \myorg@bibitem{#1}#2\par
  }{%
    \begingroup
      \color{\csname bibitem@#1\endcsname}%
      \myorg@bibitem{#1}#2\par
    \endgroup
  }%
}

\makeatother 

\title{Byzantine-Resilient Over-the-Air Federated Learning under Zero-Trust Architecture}

\author{Jiacheng~Yao,
        Wei~Shi,
        Wei~Xu,~\IEEEmembership{Fellow,~IEEE,}
        Zhaohui~Yang,~\IEEEmembership{Member,~IEEE,}
        A.~Lee~Swindlehurst,~\IEEEmembership{Fellow,~IEEE,}
        and Dusit~Niyato,~\IEEEmembership{Fellow,~IEEE}
%        % <-this % stops a space
\thanks{J. Yao, W. Shi, and W. Xu are with the National Mobile Communications Research Laboratory, Southeast University, Nanjing 210096, China, and are also with the Purple Mountain Laboratories, Nanjing 211111, China (e-mail: \{jcyao, wshi, wxu\}@seu.edu.cn).}
\thanks{Z. Yang is with the Research Institute of Intelligent Networks, Zhejiang Lab, Hangzhou 311121, China, and also with the College of Information Science and Electronic Engineering, Zhejiang University, Hangzhou 310027, Zhejiang, China (e-mail: yang\_zhaohui@zju.edu.cn).}
\thanks{A. L. Swindlehurst is with the Center for Pervasive Communications and Computing, Henry Samueli School of Engineering, University of California at Irvine, Irvine, CA 92697 USA (e-mail: swindle@uci.edu).}
\thanks{D. Niyato is with the School of Computer Science and Engineering, Nanyang Technological University, Singapore 308232 (e-mail: dniyato@ntu.edu.sg).}
}


\maketitle
\begin{abstract}
Over-the-air computation (AirComp) has emerged as an essential approach for enabling communication-efficient federated learning (FL) over wireless networks. Nonetheless, the inherent analog transmission mechanism in AirComp-based FL (AirFL) intensifies challenges posed by potential Byzantine attacks. In this paper, we propose a novel Byzantine-robust FL paradigm for over-the-air transmissions, referred to as federated learning with secure adaptive clustering (FedSAC). FedSAC aims to protect a portion of the devices from attacks through zero trust architecture (ZTA) based Byzantine identification and adaptive device clustering. By conducting a one-step convergence analysis, we theoretically characterize the convergence behavior with different device clustering mechanisms and uneven aggregation weighting factors for each device. Building upon our analytical results, we formulate a joint optimization problem for the clustering and weighting factors in each communication round. To facilitate the targeted optimization, we propose a dynamic Byzantine identification method using historical reputation based on ZTA. Furthermore, we introduce a sequential clustering method, transforming the joint optimization into a weighting optimization problem without sacrificing the optimality. To optimize the weighting, we capitalize on the penalty convex-concave procedure (P-CCP) to obtain a stationary solution. Numerical results substantiate the superiority of the proposed FedSAC over existing methods in terms of both test accuracy and convergence rate.
\end{abstract}
\begin{IEEEkeywords}
Federated learning (FL), over-the-air computation (AirComp), zero trust architecture (ZTA), Byzantine attacks.
\end{IEEEkeywords}

\section{Introduction}
\label{sec:intro}

\IEEEPARstart{T}{he} integration of communication and artificial intelligence (AI) is anticipated to be a key application scenario in the upcoming sixth-generation (6G) wireless networks \cite{itu}. Future 6G networks will benefit from AI technology while also facilitating its deployment, namely AI for 6G and 6G for AI, thus enabling ubiquitous intelligence. In particular, recent advances in federated learning (FL) have introduced a promising distributed AI paradigm over wireless networks, addressing the demands of large-scale edge intelligent applications \cite{weixu,bianji1,zhu2023pushing,gomore}.
In a wireless FL network, distributed devices work together to train a shared learning model under the guidance of a parameter server (PS). The devices involved in the training procedure iteratively update the model by exchanging model parameters instead of transferring raw data, which helps preserve data privacy because sensitive information is neither directly shared among devices nor with the PS \cite{ZTE}. However, the massive transmission of model parameters has made communication a major bottleneck in wireless FL \cite{mzchen,dva}. %zhyang,
Particularly in the uplink transmission,  a significant number of local devices are involved in transmitting local gradients to the PS for conducting model updates, which poses formidable challenges for conventional communication systems. Hence, certain prior efforts have explored selecting a subset of the devices to partake in the updates of each round, aiming to mitigate excessive communication overhead \cite{schedule,yhao}. However, this approach inevitably comes at the expense of convergence rate and accuracy.

To facilitate the deployment of FL in wireless networks, over-the-air computation (AirComp) is viewed as a key enabling technique \cite{gxzhu,imperfect,ymshi2}. In this approach, the uncoded signals at distributed devices undergo amplitude modulation and are then concurrently transmitted by reusing the available radio resources. After preprocessing at the transmitter, the automatic computation of the objective function over the air is facilitated by the superposition property of multiple access (MAC) channels, thus giving rise to the term AirComp. It seamlessly merges communication and computational functions, which aligns with the requirements of communication-efficient FL \cite{yhao2}. Harnessing the advantages of AirComp, the uplink transmission and aggregation of the model parameters in FL tasks are integrated over the air, which allows for simultaneous access by numerous devices within a limited bandwidth and a substantial reduction in both transmission and computational latency. In recent investigations, researchers have delved into the specific design of AirComp-based FL (AirFL), covering power control \cite{power2}, device scheduling \cite{scheduleAir}, and beamforming~\cite{lee}.

As FL continues to expand in applications such as smart cities, autonomous driving networks, and industrial internet of things (IoT), the involvement of a large number of edge devices poses significant security risks. Byzantine attacks, where unreliable or malicious devices disrupt the learning process, can severely compromise both the accuracy and reliability of the system \cite{Byz,trustcom,aloqaily}. What makes the situation worse is when FL is deployed over wireless networks via AirComp, where the open nature of wireless channels and nonorthogonal superimposition characteristics in AirComp exacerbate the risks \cite{qhwu}. Although it offers significant performance gains, this inherent analog transmission mechanism in AirComp is a double-edged sword.
Unlike traditional FL schemes, AirFL does not acquire individual gradients from each device. Instead, it relies on the superposition of all local gradients, which makes AirFL more vulnerable to malicious attacks \cite{shi2024empowering,airTSP}. This is especially for Byzantine attacks, where attackers can exploit this vulnerability by sending arbitrary gradient values, thereby interfering with the overall model aggregation process and hindering the convergence of the FL algorithm. To counter the Byzantine attacks, traditional approaches aim at excluding gradient outliers, such as in the Krum and Multi-Krum algorithms \cite{krum}. However, these approaches cannot be directly applied to AirFL due to the lack of access to specific gradients from individual devices. 

Given this background, some preliminary studies have considered robust design for combating Byzantine attacks under AirFL \cite{bev,ymshi,park,liye1,liye2}. In \cite{bev}, a best effort voting (BEV) power control policy was proposed to enhance robustness against Byzantine attacks, aiming at boosting the strength of useful signals using full-power transmission. While BEV enhances the fidelity of the legitimate signals, it neither directly suppresses nor filters out Byzantine interference. In addition to BEV, the AirComp-based Weiszfeld algorithm was proposed in \cite{ymshi} to obtain a smoothed geometric median aggregation approach against Byzantine attacks. In another approach, the authors in \cite{onebit} exploited a one-bit gradient quantization and hierarchical voting framework, which is theoretically proven to be robust against Byzantine devices. 
However, low-precision quantization compromises performance and majority voting is inadequate in withstanding formidable Byzantine attacks. Furthermore, to promote the application of the existing robust aggregation algorithms, multiple studies, e.g., \cite{park,liye1,liye2}, proposed to create multiple individual gradients via device clustering and hierarchical AirComp. In general, the core idea of these works is to reduce the number of honest gradients contaminated by Byzantine attackers through user grouping and techniques that exclude anomalous gradients. As shown in Fig. \ref{fig:air}, devices within each cluster aggregate to obtain a virtual gradient via AirComp, allowing multiple virtual gradients from different clusters to be received at the PS. Based on these gradients, the PS can apply existing robust aggregation algorithms to effectively exclude abnormal gradients, ensuring enhanced resilience against Byzantine devices. For instance, the authors of \cite{park} suggested incorporating a reference gradient, computed directly by the PS, to effectively identify and filter out malicious gradients. Moreover, the geometric median aggregation method was adopted in \cite{liye1,liye2} for robust aggregation with a theoretical convergence guarantee. However, existing works primarily focused on adopting various robust aggregation algorithms while overlooking the significance of device clustering design, which was often done randomly. This oversight limits their effectiveness in enhancing robustness against Byzantine attacks. Furthermore, the absence of dynamic authentication in all robust schemes mentioned above prevents active suppression of Byzantine attacks, necessitating passive defense strategies instead. This limitation also significantly constrains the FL performance.

%Furthermore, to promote the application of the existing robust aggregation algorithms, multiple studies, e.g., \cite{park,liye1,liye2}, proposed to create multiple individual gradients via user clustering and hierarchical AirComp. The core idea of clustering is to reduce the number of honest gradients contaminated by Byzantine attackers through user grouping and techniques that exclude anomalous gradients. However, random clustering methods used in existing studies rarely considered specific characteristics of the wireless environment and the objectives of an FL task. This oversight limits their effectiveness in enhancing robustness against Byzantine attacks. Furthermore, the absence of dynamic authentication in all robust schemes mentioned above prevents active suppression of Byzantine attacks, necessitating passive defense strategies instead. This limitation also significantly constrains the FL performance.

\begin{figure}[!t]
  \centering
  \centerline{\includegraphics[width=3.2in]{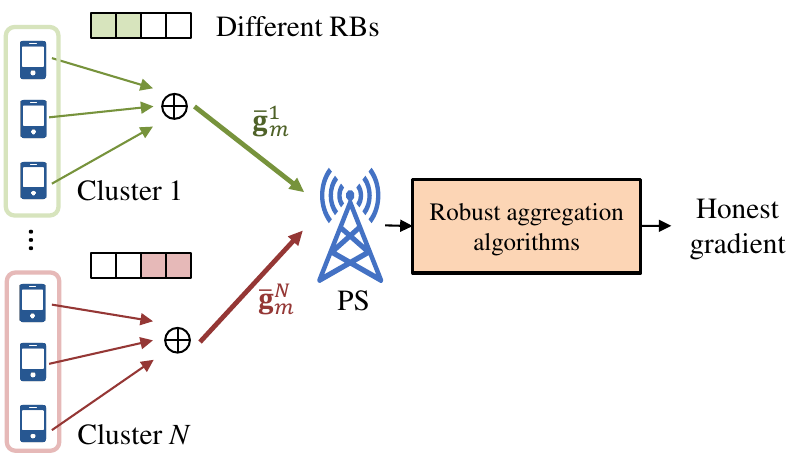}}
  \caption{Architecture of existing robust aggregation methods via hierarchical AirComp and device clustering.}\label{fig:air}
\end{figure}

To address the aforementioned challenges, zero trust architecture (ZTA), characterized by its core principle of ``never trust and always verify," is widely acknowledged as a revolutionary approach to bolster security in 6G networks \cite{ZTAo,zta1,bianji2}. Specifically, in ZTA, no access device is deemed inherently trustworthy. Instead, it relies on continuous trust evaluation and authentication mechanisms to identify and mitigate attacks \cite{zta2}. The core principle of ZTA aligns well with the requirements of AirFL to address the Byzantine attacks. Hence, building upon ZTA, we explore an enhanced Byzantine-robust design for AirFL to further improve the convergence performance and security in this paper. We refer to our approach as federated learning with secure adaptive clustering (FedSAC). The main contributions of our work are summarized as follows:
\begin{itemize}
    \item We establish a universal framework of clustering-based robust aggregation method in which the devices are partitioned into distinct clusters, and AirComp is performed separately within each cluster. Meanwhile, each cluster uses different time/frequency resources for transmission, thus generating multiple gradient values from each cluster to support various existing robust aggregation algorithms. Convergence of the robust aggregation scheme is assured in theory and the impacts of the clustering and the weighting factors assigned to each device in aggregation are also quantitatively captured. 
    \item Based on the analytical results, we formulate a long-term joint optimization problem for the clustering and weighting allocation to enhance the convergence performance. This problem is decomposed into subproblems in~each communication round using the Lyapunov~method. To further streamline the proposed joint optimization problem, we introduce the ZTA-empowered FedSAC framework, which integrates Byzantine identification, and the joint design of the clustering and weighting allocation. Specifically, we first employ a continuous Byzantine identification scheme for each round based on historical reputation, which serves as a prerequisite for tailoring the design of the clustering and weighting factors against Byzantine devices. For clustering optimization, we introduce a sequential clustering scheme based on equivalent channel conditions under specified weighting factors and justify its optimality. By incorporating the sequential clustering into the original problem, we transform the joint optimization problem into an equivalent optimization of the weighting factors. Since the equivalent weighting allocation problem involves mixed integer nonlinear programming (MINLP), we leverage the penalty convex-concave procedure (P-CCP) to reformulate the integer constraints, fostering the realization of a stable solution.
    \item We demonstrate the effectiveness of the proposed FedSAC approach via numerical results. In particular, we have significantly accelerated the convergence, attaining performance levels that approach attack-free scenarios. Compared with existing random clustering-based methods, the proposed FedSAC approach is shown to achieve remarkable robustness across diverse signal-to-noise ratios (SNRs) and Byzantine attack configurations. Moreover, the proposed approach achieves good performance with only a small number of clusters, thereby minimizing wireless resource consumption and maximizing the benefits of AirComp.
\end{itemize}

The rest of this paper is organized as follows. In Section~\Rmnum{2}, we describe the AirFL and Byzantine attack models. Section \Rmnum{3} introduces the clustering-based robust aggregation method and formulates the joint optimization problem based on a one-step convergence analysis. In Section~\Rmnum{4}, we outline the general framework of the proposed ZTA-empowered FedSAC approach. Simulation results and conclusions are given in Sections \Rmnum{5} and \Rmnum{6}, respectively. 

\emph{Notation:} Boldface lowercase (uppercase) letters represent vectors (matrices). The set of real numbers is denoted by $\mathbb{R}$. The set $[N]$ is equal to $\{1,2,\cdots,N\}$. The superscripts $(\cdot)^T$ and $(\cdot)^\ast$ stand for the transpose and conjugate operations, respectively. The symbol $\Re(\cdot)$ returns the real part of a complex number. The operator $|\cdot|$ returns the absolute value when the input is a complex scalar or the number of elements when the input is a set. The operator $\left\|\cdot\right\|$ denotes the Euclidean norm of a vector. A circularly symmetric complex Gaussian distribution is denoted by ${\cal {CN}}$, and $\mathbb{E}[\cdot]$ denotes the expectation. Moreover, the summaries of symbols are presented in Table \ref{NOTATION}.

\section{System Model}
\label{sec:sys}
We consider a wireless FL framework in which $K$ distributed devices jointly train a shared machine learning model under the orchestration of a central PS. We note that the PS is always trusted during the training process. All devices use their local datasets $\{\mathcal{D}_k\}_{k=1}^K$ to train the model parameters, $\mathbf{w}\in \mathbb{R}^d$, by minimizing a global loss function defined as 
\begin{align}\label{e1}
F(\mathbf{w})=\sum_{k=1}^K \frac{|\mathcal{D}_k|}{\sum_{i=1}^K |\mathcal{D}_i|}  F_k(\mathbf{w},\mathcal{D}_k).
\end{align}
In (\ref{e1}), $F_k(\mathbf{w},\mathcal{D}_k)$ is the local loss function at device $k$. It is formulated as
\begin{align}
F_k(\mathbf{w},\mathcal{D}_k)=\frac{1}{\left \vert \mathcal{D}_k\right \vert}\sum_{\mathbf{u}\in \mathcal{D}_k} \mathcal{L}(\mathbf{w},\mathbf{u}),
\end{align}
where $\mathbf{u}$ denotes the data sample and $\mathcal{L}(\mathbf{w},\mathbf{u})$ represents the sample-wise loss function. Given the system's heterogeneity, we consider that local datasets across different devices are non-independent and non-identically distributed (non-IID). Without loss of generality, we assume that the sizes of all local datasets are the same, i.e., each device is expected to play an equal role in model training, and an extension for different sizes of local datasets is straightforward.

\begin{table}[t]
    \caption{Summary of Main Symbols}
    \begin{center}
        \begin{tabular}{|l|p{6cm}|}
            \hline
            Notation & Definition\\
            \hline
            $\mathcal{D}_k$ & Local dataset of device $k$\\
            $F(\mathbf{w})$ & Global loss function\\
            $F_k(\mathbf{w},\mathcal{D}_k)$ & Local loss function of device $k$\\
            $\mathbf{w}_m$ & Model parameter at round $m$\\
            $\mathbf{g}_{m}^k$; $\bar{\mathbf{g}}_{m}^n$ & Local gradient of device $k$ at round $m$; aggregated gradient of cluster $n$ at round $m$\\
            $\alpha_{m,k}$ & Weighting factor assigned to device $k$ at round $m$\\
            $\rho_{m,k}$ & Preprocessing factor for device $k$ at round $m$\\
            $\zeta_{m,n}$ & Scaling factor for cluster $n$ at round $m$\\
            $\gamma_{m,k}$; $\bar{\gamma}_{m,k}$ & Contribution (normalized) of device $k$ at round $m$\\
            $r_{m,k}$ & Reputation of device $k$ at round $m$\\
            $q_{m,k}$ & Virtual queue for fairness constraint\\
            $\eta$ & Learning rate\\
            $L$ & Lipschitz parameter of $F_k(\cdot)$\\
            $\mu$ & Constant for Polyak-{\L}ojasiewicz Inequality\\
            $\delta_k$ & Gradient divergence of device $k$\\
            $\gamma_{\mathrm{th}}$ &  Threshold for channel truncation\\
            $b$ &  Expected long-term average contribution\\
            $\alpha$ &  Penalty parameter for absence\\
            $\mathcal{A}_m$ &  Set of the activated devices at round $m$\\
            $\mathcal{B}_m$ & Set of clusters passing the inspection at round $m$\\
            $\mathcal{C}_{m,n}$ & Set of cluster $n$ at round $m$\\
            $\mathcal{D}_m$ & Set of activated and honest devices at round $m$\\
            $\mathcal{S}_m$ &  Set of devices with effective participation at round~$m$\\
            \hline
        \end{tabular}
        \label{NOTATION}
    \end{center}
\end{table}

\subsection{AirFL Model}

To minimize the loss function in (\ref{e1}), we apply the widely used FedSGD algorithm. In addition, AirComp is exploited for efficient transmission and model aggregation.  Specifically, in the $m$-th round, the PS first broadcasts the up-to-date global model $\mathbf{w}_{m}$ to all devices. Each device $k$ initializes its local model as $\mathbf{w}_{m}$ and computes its local stochastic gradient as 
\begin{align}\label{e3}
    \mathbf{g}_m^k =\nabla F_k (\mathbf{w}_m,\mathbf{\xi}_m^k),
\end{align}
where $\mathbf{\xi}_{m}^k$ represents the mini-batch selected from $\mathcal{D}_k$ with size $b$. Then, all devices report their local gradients to the PS for further aggregation. Using AirComp, the procedures of transmission and computation are integrated as
\begin{align}\label{eq4}
    \mathbf{y}_m= \sum_{k=1}^K h_{m,k} \beta_k \rho_{m,k} \mathbf{g}_{m}^k +\mathbf{z},
\end{align}
where $h_{m,k}\sim \mathcal{CN}(0,1)$ is the small-scale fading between the PS and device $k$ in the $m$-th round, $\beta_k$ is the large-scale fading, $\rho_{m,k}$ represents the pre-processing factor at device $k$ for computation tasks, and $\mathbf{z}\sim \mathcal{CN}(\mathbf{0},\sigma^2\mathbf{I})$ denotes additive Gaussian noise with zero mean and variance $\sigma^2$. Thanks to the analog aggregation in (\ref{eq4}), aggregated local gradients in AirComp prevent the PS from identifying individual device gradients, thus reducing the potential of privacy leakage \cite{leakage}. In this sense, even if the PS is semi-honest, the unique advantages of AirComp limit privacy risks.

According to \cite{gxzhu}, the pre-processing factor is set as
\begin{align}\label{eq6}
\rho_{m,k} =\left \{  
\begin{array}{ll}
\frac{\zeta_{m}\alpha_{m,k}}{h_{m,k}\beta_k},& \vert h_{m,k}\vert \geq \gamma_{\mathrm{th}},\\
0,&\vert h_{m,k}\vert<\gamma_{\mathrm{th}},
\end{array}
\right.
\end{align}
where $\gamma_{\mathrm{th}}$ is a predetermined threshold to avoid a catastrophic impact of deep fading\footnote{In AirFL, the model upload is synchronized for all activated devices and hence stragglers can be avoided. Also, the probability of being truncated is the same for all devices, thus no fairness issues arise.}, $\zeta_{m}$ is a scaling factor for ensuring the transmit power constraint, and $\alpha_{m,k}$ denotes the weighting factor assigned to device $k$ at round $m$. The weighting factors are normalized such that $\sum_{k=1}^K \alpha_{m,k}=1$. In other approaches, all devices are allocated the same weight, i.e., $\alpha_{m,k}=\frac{1}{K}$, for equal contribution. Here, we consider a more general case with arbitrary weighting factors to achieve further performance improvement \cite{adaweight}. With the pre-processing in (\ref{eq6}), devices with poor channel conditions are truncated, ensuring a large scaling factor, $\zeta_m$, which amplifies transmitted signals to better combat additive noise.

We denote the set of the activated devices at the $m$-th round by 
\begin{align}
    \mathcal{A}_m=\left \{k\left \vert \left\vert h_{m,k}\right\vert  \geq \gamma_{\mathrm{th}}\right.\right\}.
\end{align}
For simplicity, we consider a uniform transmit power budget $P_{\max}$ for each device and choose the factor $\zeta_m$ to guarantee $\left \Vert \rho_{m,k}\mathbf{g}_m^k \right \Vert^2 \leq P_{\max}$.
We assume that the local gradients are uniformly bounded by a finite value $G$ \cite{jundu}, i.e., $\left \Vert \mathbf{g}_m^k \right \Vert \leq G$, and hence the factor $\zeta_m$ is selected as
\begin{align}
\zeta_m = \frac{\sqrt{P_{\max}}}{G}\min_{k\in\mathcal{A}_m}\left\{ \frac{\vert h_{m,k}\vert \beta_k}{\alpha_{m,k}}\right \}.
\end{align}
At the receiver, the PS updates the global model as
\begin{align}\label{eq7}
\mathbf{w}_{m+1}&=\mathbf{w}_{m}-\eta \frac{\Re\{\mathbf{y}_m\}}{\zeta_m}\nonumber \\
&=\mathbf{w}_{m}-\eta \sum_{k\in \mathcal{A}_m}\alpha_{m,k} \mathbf{g}_m^k+ \bar{\mathbf{z}},
\end{align}
where $\eta$ denotes the learning rate, and $\bar{\mathbf{z}}$ is the equivalent noise with distribution $\mathcal{N}(\mathbf{0},\frac{\eta^2 \sigma^2}{2\zeta_m^2}\mathbf{I})$. 
The FL algorithm iterates the above procedure until convergence and achieves an effective global model.

\subsection{Byzantine Attack Model}
In this paper, we assume that model poisoning attacks are performed by Byzantine attackers based on the characteristics of AirFL \cite{bev,lee,liye1,liye2}. Specifically, each attacker actively manipulates the model updates by sending arbitrarily malicious gradients to the PS, thereby causing model training failure and FL divergence. In addition, there are $M<K$ Byzantine devices and $K-M$ honest devices, and the message reported by device $k$ in (\ref{eq4}) is replaced by 
\begin{align}
    \mathbf{g}_m^k = \left \{ 
    \begin{array}{ll}
        \nabla F_k(\mathbf{w}_m ,\xi_m^k), &\text{if $k$-th device is honest},  \\
        *, & \text{otherwise},
    \end{array}
    \right.
\end{align}
where $*$ denotes an arbitrary value.  Unlike wireless FL algorithms implemented using classical orthogonal access methods, the characteristic of non-orthogonal access in AirComp facilitates the concealment of Byzantine devices among the honest ones, significantly aggravating the challenge of identifying and mitigating Byzantine attacks. Therefore, the design of tailored AirFL frameworks and corresponding robust algorithms becomes imperative to mitigate potential adverse effects of malicious attacks. In the next section, we develop a general framework for robust aggregation in AirFL to combat the Byzantine attacks.

\section{Clustering-based Robust Aggregation Scheme for AirFL}

To reduce a negative impact of malicious devices, we exploit the hierarchical AirComp method shown in Fig. \ref{fig:air}, in which all $K$ devices are uniformly divided into $N$ clusters and all clusters participate in the FL procedure with distinct time/frequency resource blocks (RBs). Under this framework, we introduce multiple local updates for AirFL, which helps support existing Byzantine-robust aggregation methods. In this section, we first elucidate the particulars of the robust aggregation scheme, subsequently formulating an additional optimization problem for the clustering based on our theoretical convergence analysis.

\subsection{Scheme Design}
Within a specific cluster $n$ at round $m$, denoting $\mathcal{C}_{m,n}$, a step similar to (\ref{eq4}) is adopted. Then the aggregated local update for the cluster $\mathcal{C}_{m,n}$ is given by
\begin{align}
    \bar{\mathbf{g}}_m^n =\sum_{k\in \mathcal{C}_{m,n}\cap \mathcal{A}_m} \alpha_{m,k} \mathbf{g}_{m}^k+ \mathbf{z}_n,
\end{align}
where $\mathbf{z}_n\sim \mathcal{N}(\mathbf{0},\frac{ \sigma^2}{2\zeta_{m,n}^2}\mathbf{I})$ and
\begin{align}\label{eq10}
    \zeta_{m,n}=\frac{\sqrt{P_{\max}}}{G}\min_{k\in \mathcal{C}_{m,n} \cap \mathcal{A}_m}\left\{ \frac{\vert h_{m,k}\vert \beta_k}{\alpha_{m,k}}\right \}.
\end{align}
Through hierarchical AirComp, we construct $N$ virtual local gradients, i.e., $\{\bar{\mathbf{g}}_m^n\}_{n=1}^N$, corresponding to the $N$ clusters.

Building upon the virtual local gradients, we implement some widely used robust aggregation methods. In this paper, we adopt the idea of the classical FLTrust scheme in \cite{fltrust}, which is also consistent with that in \cite{tsp2}. In this approach, the PS collects a small and clean root training dataset and generates its own update, $\mathbf{g}_m^0$, based on the root dataset. Then, upon receiving the aggregated local updates $\{\bar{\mathbf{g}}_m^n\}_{n=1}^N$, the PS exploits $\mathbf{g}_m^0$ to identify contaminated updates via cosine similarity \cite{cosine} and excludes them from the aggregation process. In particular, the local update $\bar{\mathbf{g}}_m^n$ is excluded if the cosine similarity between $\bar{\mathbf{g}}_m^n$ and $\mathbf{g}_m^0$ is less than a predetermined threshold $\varrho$.
%\begin{align}
 %   \frac{<\bar{\mathbf{g}}_m^n,\mathbf{g}_m^0>}{\Vert \bar{\mathbf{g}}_m^n \Vert \cdot \Vert \mathbf{g}_m^0 \Vert}\leq \varrho,
%\end{align}
By selecting an appropriate $\varrho$, we can strike a balance between robustness against Byzantium and unbiasedness of the global model. 
It is noteworthy that the robust aggregation method employed in our scheme is capable of effectively handling any number of Byzantine devices and exhibits commendable robustness \cite{tsp2}. Apart from the considered method, other state-of-art Byzantine-robust methods, e.g., the median-based algorithms in \cite{tspbay,trimean}, can also be directly applied. Nonetheless, a prerequisite for the efficacy of these Byzantine robust algorithms is $M<\frac{N}{2}$, emphasizing the need for a sufficient number of clusters. This requirement unavoidably increases resource consumption and significantly diminishes the advantages offered by AirComp. Hence, without loss of generality, the subsequent optimization design will be based on the adopted robust aggregation method.

After eliminating the contaminated updates, the model update in (\ref{eq7}) is modified as
\begin{align}\label{eq11}
    \mathbf{w}_{m+1}&=\mathbf{w}_{m}-\eta \sum_{n\in \mathcal{B}_m}\bar{\mathbf{g}}_m^n \nonumber \\
    &=\mathbf{w}_m -\eta \sum_{k\in\mathcal{S}_m} \alpha_{m,k}\mathbf{g}_m^k - \eta\sum_{n\in \mathcal{B}_m} \mathbf{z}_n,
\end{align}
where $\mathcal{B}_m$ represents the set of clusters that have passed the inspection, and $\mathcal{S}_m$ denotes the set of devices successfully participating in the aggregation at round $m$. The set $\mathcal{S}_m$ is given by
    $\mathcal{S}_m = \bigcup_{n\in \mathcal{B}_m} \left(\mathcal{C}_{m,n}\cap \mathcal{A}_m\right)$.

\subsection{Convergence Analysis}
To facilitate convergence analysis under the robust aggregation, we make some popular assumptions on the loss functions, which are used for example in~\cite{mzchen,gxzhu,optimized} and entail no loss of generality.

\emph{Assumption 1}: The local loss functions $F_k(\cdot)$ are differentiable and have $L$-Lipschitz gradients, which implies
\begin{align}
F_k(\mathbf{w})\leq  F_k(\mathbf{v}) +\nabla F_k(\mathbf{v})^T (\mathbf{w}-\mathbf{v})+\frac{L}{2}\Vert \mathbf{w}-\mathbf{v}\Vert^2.
\end{align}

\emph{Assumption 2}: The divergence between the stochastic  and global gradients is bounded, i.e., $\left\Vert\mathbf{g}_m^k-\nabla F(\mathbf{w}_m) \right \Vert\leq \delta_k$.

\emph{Assumption 3} (Polyak-{\L}ojasiewicz Inequality): We denote the minimum value of the global loss function as $F^*$. There exists a constant $\mu>0$ such that the global loss function satisifies
\begin{align}
    \left\Vert \nabla F(\mathbf{w})\right \Vert^2 \geq 2\mu \left(F(\mathbf{w})-F^*\right ).
\end{align}

%The stochastic gradient is unbiased and variance-bounded, i.e., $\mathbb{E}\!\left [ \mathbf{g}_m^k \right]\!=\!\nabla F_k (\mathbf{w}_m)$, and $\mathbb{E}\left[\left \Vert \mathbf{g}_m^k-\nabla F_k(\mathbf{w}_m)\right \Vert^2\right] \leq \chi^2$.

Based on the above assumptions and the model update in (\ref{eq11}), we analyze the convergence of a specific communication round in the following lemma. To explicitly analyze the convergence behavior, we define the expected reduction in the loss function at round $m$ by $g_m\triangleq \mathbb{E}\left[F(\mathbf{w}_{m+1})\right]-F(\mathbf{w}_m)$, where the expectation is taken over the additive noise. 

\begin{figure*}
\begin{align}\label{eq14}
    g_m \leq& -\frac{\eta}{2}\left (\sum_{k\in \mathcal{S}_m}\alpha_{m,k}\right )\left \Vert \nabla F(\mathbf{w}_m)\right \Vert^2  -\frac{\eta\!-\!L\eta^2}{2}\sum_{k\in\mathcal{S}_m}\alpha_{m,k} \left(\left \Vert  \mathbf{g}_m^k \right \Vert^2
    -\frac{\eta}{\eta-L\eta^2}\delta_k^2\right)+\frac{L\eta^2}{2}\sigma_m^2\triangleq \bar{g}_m,
\end{align}
\hrulefill
\end{figure*}

\begin{lemma}\label{lemma1}
The expected reduction in the loss function for the considered robust aggregation scheme at round $m$,  denoted by $g_m$, is bounded by (\ref{eq14}) at the top of the next page, where $\sigma_m^2\triangleq \sum_{n\in\mathcal{B}_m} \frac{ \sigma^2}{2\zeta_{m,n}^2}$.
\end{lemma}
\begin{proof}
    Please refer to Appendix \ref{app_lemma1}. \hfill $\square$
\end{proof}

\begin{remark}
    Defining 
    \begin{align}
        \gamma_{m,k}\triangleq \Vert \mathbf{g}_m^k \Vert^2 -\frac{\eta}{\eta-L\eta^2}\delta_k^2,
    \end{align} 
    we theoretically characterize the impact of incorporating device $k$ into the $m$-th round of FL training, which can be interpreted as the contribution of device $k$ in the $m$-th round. The contribution is scaled by the weighting factor $\alpha_{m,k}$. 
    For honest devices, the parameter $\delta_k$ is contingent solely upon the degree of non-IID local datasets, with its magnitude being constrained. Consequently, the participation of honest devices yields positive benefits. Conversely, for Byzantine devices, the local gradient can be any value, leading to a substantial deviation from the global gradient and, consequently, a negative contribution. The definition of contribution aligns with the intuitive understanding presented in existing studies \cite{schedule,jundu}. Devices exhibiting larger local gradient values exert a more significant influence on global model training compared with those with smaller values. Hence, the magnitude of $\Vert \mathbf{g}_m^k\Vert$ commonly serves as an importance indicator for device scheduling \cite{schedule}. The derived theoretical value of the contribution incorporates the impact of~data heterogeneity, constituting a more universally applicable characterization.
\end{remark}

\begin{remark}
    By exploiting the robust aggregation method, with a sufficiently small learning rate, the convergence of the FL algorithm is always guaranteed. Concretely, we can exclude gradients that deviate significantly from the global gradient, thereby ensuring a positive contribution $\gamma_{m,k}$ from participating devices in $\mathcal{S}_m$. The noise term is linearly related to $\eta^2$ and hence its effect is a higher-order term, which can be safely ignored as $\eta\to 0$. Therefore, with these assumptions, we conclude that $g_m\leq \bar{g}_m<0$, indicating that the loss function is monotonically decreasing with the number of communication rounds. Given that the objective of FL is to minimize the loss function with a lower bound, achieving final convergence is guaranteed.
\end{remark}

%Based on \emph{Lemma \ref{lemma1}} and the robust aggregation method, we have the following theorem regarding convergence.

%\begin{theorem}\label{theorem0}
   % Via robust aggregation, there always exists a sufficiently small learning rate to ensure the convergence of the FL algorithm.
%\end{theorem}

%\begin{proof}
  %  We can exclude gradients that deviate significantly from the global gradient, thereby ensuring a positive contribution $\gamma_{m,k}$ from participating devices in $\mathcal{S}_m$. The noise term is linearly related to $\eta^2$ and hence its effect is a higher-order term, which can be safely ignored as $\eta\to 0$. Therefore, with these assumptions, we conclude that $g_m\leq \bar{g}_m<0$, indicating that the loss function monotonically decreases with the number of communication rounds. Given that the objective of FL is to minimize the loss function with a lower bound, achieving final convergence is guaranteed. \hfill $\square$
%\end{proof}

While existing clustering-based robust aggregation schemes effectively ensure convergence, these approaches remain a passive defense against Byzantine attacks. They lack the capability for identification and targeted clustering optimization specifically tailored for Byzantine devices, thereby degrading the convergence rate and achievable performance. This issue will be further explored in the next subsection.

\subsection{Optimization Problem Formulation}
To enhance convergence, it is imperative to optimize the algorithm design in each round with the aim of minimizing the value of the loss function. For a tractable objective, we exploit the definition of $g_m$ and reformulate the expected value of the loss function in the $M$-th round as 
\begin{align}\label{eq18}
    \mathbb{E}[f(\mathbf{w}_M)] =\sum_{m=0}^{M-1}g_m +F(\mathbf{w}_0).
\end{align}
As $F(\mathbf{w}_0)$ in (\ref{eq18}) is a constant, it is therefore equivalent to minimize the sum $\sum_{m=0}^{M-1}g_m$. Since it is challenging to accurately obtain the actual value of $g_m$ at the PS, we resort to utilizing the upper bound, $\bar{g}_m$, as a surrogate. Examination of $\bar{g}_m$ in (\ref{eq14}) reveals that it primarily comprises two components, the positive contribution resulting from the participation of honest devices and the negative impact of additive noise. Both aspects are intricately linked to the specific choice of clustering and weighting factors. Accordingly, we formulate the  joint optimization problem:
\begin{alignat}{2}
    (\text{P1})\enspace\! & \mathop{\text{minimize}}_{\mathcal{U}_1,\cdots,\mathcal{U}_\infty} \quad \lim_{M\to \infty} \sum_{m=1}^M \bar{g}_m \nonumber \\
    &\text{subject to} \quad \text{C}_1:\sum_{k=1}^K \alpha_{m,k}=1,\enspace \forall m\nonumber \\
    &\quad \quad \quad \quad \quad \text{C}_2: \left | \mathcal{C}_{m,n} \right | =\bar{K}, \enspace n\in [N],\enspace \forall m\nonumber \\
    &\quad \quad \quad \quad \quad \text{C}_3:\! \lim_{M\!\to\infty} \!\sum_{m=1}^{M\!-\!1}\!\frac{\bar{\gamma}_{m,k}\alpha_{m,k}}{M} \! \geq \!b,  k\!\in\! [K],
\end{alignat}
where $\mathcal{U}_m \triangleq \left\{ \{\mathcal{C}_{m,n}\}_{n\in[N]}, \{\alpha_{m,k}\}_{k\in[K]} \right \}$ denotes the set of optimization variables for the $m$-th round, $\bar{K}\triangleq K/N$ represents the number of devices in each cluster and it is assumed for simplicity to be an integer, $\bar{\gamma}_{m,k}\triangleq \frac{\gamma_{m,k}}{\sum_{i=1}^K \gamma_{i,k}}$ denotes the normalized contribution of device $k$, and $b$ denotes the target average contribution of each device. The introduction of variable weighting factors in each round may compromise the fairness among devices, so we add a fairness constraint $\text{C}_3$ that is set to ensure that the long-term average contribution of each device is at least greater than $b$~\cite{fair}. The aforementioned problem in (P1) primarily encounters two challenges. Firstly, to minimize $\bar{g}_m$, accurate identification of Byzantine devices is crucial as prior information, which is unavailable at the PS. Secondly, it entails long-term constraints, introducing causal issues into instantaneous optimization. 

%As for the long-term optimization, note that apart from the constraint $\text{C}_3$, the problem in (P1) is time decoupled. 

To enable an independent optimization in each communication round, we first focus on reformulation of the long-term constraint $\text{C}_3$ by exploiting the well-known Lyapunov method \cite{fair}. In particular, we introduce a virtual queue $q_{m,k}$ for each device $k$ to measure the gap between the cumulative contribution up to round $m$ and the fairness target:
\begin{align} \label{virtualq}
    q_{m,k}=\max\left\{ q_{m-1,k}+b-\bar{\gamma}_{m-1,k}\alpha_{m-1,k},0\right\}.
\end{align}
Next, based on the drift-plus-penalty algorithm of Lyapunov optimization, the original problem in (P1) is divided into separate subproblems in each communication round. The $m$-th subproblem is:
\begin{align}
    (\text{P2})\enspace\mathop{\text{maximize}}\limits_{\{\mathcal{C}_{m,n}\}_{n\in[N]}, \{\alpha_{m,k}\}_{k\in[K]}} \quad & V \bar{g}_{m}+\sum_{k=1}^K q_{m,k}\bar{\gamma}_{m,k}\alpha_{m,k}\nonumber \\
    \text{subject to}\quad &\text{C}_1,\text{C}_2,
\end{align}
where $V$ is a penalty factor set for balancing $\bar{g}_{m}$ and the fairness constraint. In the forthcoming section, we discuss the identification of Byzantine devices and then solve problem $(\text{P2})$.

\section{Proposed Byzantine-Robust Scheme with Adaptive Clustering}
\label{sec:safe}
To effectively solve problem $(\text{P2})$, we first introduce a ZTA-based Byzantine identification method by examining the reputation of the devices in each round. Then, we decouple the clustering and weighting factor allocation. For clustering, we propose a sequential clustering method and illustrate its optimality for a given set of weighting factors. Substituting the optimal sequential clustering into $(\text{P2})$ and using the knowledge about identified Byzantine devices, we rewrite the optimization of the weighting factors, which is then solved by adopting the P-CCP algorithm.

\subsection{ZTA-Based Byzantine Identification Method}

%Prior to addressing the optimization problem, acquiring information about the Byzantine devices is imperative to facilitate the targeted design interventions. 
Given that no participating device is inherently trustworthy, continuous trust evaluation and authentication are required throughout the training process in ZTA. Based on these, the key of zero-trust security lies in establishing trust for each entity within the network, enabling the system to enforce differentiated policies for accessing network resources \cite{ZTAo,TIFS}.
Hence, we continuously identify the Byzantine devices to enable targeted interventions in problem (P2). To achieve this goal, we  employ statistical information derived from historical participation in each communication round for Byzantine identification, a technique used in \cite{byzid1,byzid2}. To this end, we first define the \emph{reputation} of device $k$ at round $m$ by
 \begin{align}\label{reputation}
    r_{m,k} \triangleq  \sum_{t=1}^{m-1} \alpha_{t,k} \bar{\gamma}_{t,k} \mathcal{J}_{t,k},
\end{align}
where 
\begin{align}\label{reputation2}
    \mathcal{J}_{t,k} \!=\!\left\{ 
    \begin{array}{ll}
    1, & \text{device $k$ participates in the  $t$-th aggregation},\\
    -\alpha, & \text{device $k$ is exluced in the  $t$-th aggregation},
    \end{array}\right.
\end{align}
and $\alpha>0$ is a predetermined penalty \cite{reputation}. The formulation in (\ref{reputation})-(\ref{reputation2}) provides a device with positive feedback for reputation if it successfully participates in training. Otherwise it receives negative feedback. Byzantine devices are more likely than honest devices to be excluded from aggregation by robust algorithms, thereby resulting in lower reputation.  Hence, we identify the $M$ devices with the lowest reputation as the Byzantine devices in the current round. 

Accurate Byzantine identification, tailored to different attack strategies, can be achieved by selecting a sufficiently large $\alpha$, thus preventing the occurrence of error traps. 
%In accordance with the selection strategy for the weighting factors, devices identified as Byzantine have their weighting factor set to 0, causing their reputations frozen. Furthermore, in cases of Byzantine devices with missed detection, a sufficiently large $\alpha$ allows a significant reputation reduction, facilitating the Byzantine identification in the future rounds. 
To support this claim, we define the probability of excluding clusters with Byzantine devices as $p$ and the probability of excluding clusters without Byzantine devices as $q$, with the reasonable assumption that $p>q$ and $q\to0$ in practice\footnote{Although Byzantine behavior may be dynamic and adaptive, the probability of a Byzantine device being excluded over the long term is still higher than that of an honest device, particularly with latest robust aggregation methods like \cite{qdong,ytao}. Otherwise, if it behaves too benignly, it should no longer be regarded as a Byzantine device.}. If a misidentification occurs at round $m$, meaning that at least one honest device $i$ is mistakenly classified as Byzantine and one actual Byzantine device $j$ is missed, their reputations then follow $r_{m,i}<r_{m,j}$. According to the optimal selection strategy for the weighting factors discussed subsequently, devices identified as Byzantine have their weighting factor set to 0, causing their reputations frozen. Then, if we select $\alpha$ satisfying $\frac{1-p}{p}<\alpha\leq\frac{1-q}{q}$, it follows
\begin{align}
    &\mathbb{E}[r_{m^\prime+1,j}]-r_{m^\prime,j}<0, \nonumber \\
    &\mathbb{E}[r_{m^\prime+1,j}]-r_{m^\prime,j}\leq \mathbb{E}[r_{m^\prime+1,v}]-r_{m^\prime,v}  ,\enspace \forall v\neq j,
\end{align}
for any round $m^\prime \geq m$ with Byzantine device $j$ being misidentified. This implies that the reputation of device $j$ is expected to monotonically decrease at the fastest rate until it is eventually identified as a Byzantine device. By repeating this analysis, we conclude that all Byzantine devices are correctly identified in the steady state regardless of the specific type of attack.
Hence, the long-term trust assessment based on ZTA can accurately identify Byzantine devices from historical information.

For faster and accurate identification of Byzantine devices, we need to fine-tune $\alpha$ to adapt to different attack types and parameter settings. In practical deployment, adaptive clustering is not performed during the initial rounds; instead, we focus solely on updating the reputations. This serves two purposes: First, accumulating statistics over multiple rounds provides a more accurate initial identification of Byzantine devices. Second, the statistical results allow us to estimate $p$, which in turn facilitates the empirical selection of $\alpha$.

Leveraging results of Byzantine identification as prior information, we conduct targeted optimization in the following subsections.

\subsection{Clustering Optimization}
Given the optimal weighting factors, we focus on the clustering optimization. Removing irrelevant constants, the problem in $(\text{P2})$ reduces to
\begin{align}
    (\text{P3})\enspace \mathop{\text{maximize}}_{\{\mathcal{C}_{m,n}\}_{n=1}^N} &\quad \sum_{k\in\mathcal{S}_m} v_{m,k} -L\eta \sigma_m^2\nonumber \\
    \text{subject to}&\quad \text{C}_2,
\end{align}
where $v_{m,k}=\alpha_{m,k}\left( \left \Vert \nabla F(\mathbf{w}_m)\right \Vert^2 +(1-L\eta)\left \Vert  \mathbf{g}_m^k \right \Vert^2-\delta_k^2\right)$ is constant. It is evident that the clustering optimization must address the dual objectives: 1) reduce noise, and 2) augment the sum contribution of participating devices.

To begin with, the equivalent noise power is calculated as 
\begin{align}\label{eq17}
    \sigma_{m}^2 \!=\!\sum_{n\in \mathcal{B}_m} \!\frac{ \sigma^2}{2\zeta_{m,n}^2}\!=\!\sum_{n\in \mathcal{B}_m} \!\frac{ \sigma^2G^2}{2P_{\max}}\! \max_{k\in \mathcal{C}_{m,n}}\!\left\{ \frac{\alpha_{m,k}^2}{\vert h_{m,k}\vert^2 \beta_k^2}\right \}.
\end{align}
From the perspective of minimizing the equivalent noise power, we construct the following sequential clustering method. We first sort the equivalent channel conditions of all users as follows $\frac{\vert h_{m,1^\prime}\vert \beta_{1^\prime}}{\alpha_{m,1^\prime}}\leq \frac{\vert h_{m,2^\prime}\vert \beta_{2^\prime}}{\alpha_{m,2^\prime}}\leq \cdots \leq\frac{\vert h_{m,K^\prime}\vert \beta_{K^\prime}}{\alpha_{m,K^\prime}}$. 
Subsequently, we allocate the $\bar{K}$ devices with poorest channel conditions to cluster $1$. We proceed by grouping the next $\bar{K}$ devices with the worst channel conditions among the remaining users into a cluster, continuing this process until all devices are assigned. For given weighting factors $\{\alpha_{m,k}\}_{k\in[K]}$, we denote the clustering results obtained via sequential clustering as $\mathcal{C}_{m,n}\left(\{\alpha_{m,k}\}\right)$ for $n\in [N]$, which are a function of the weighting allocation results.
The following theorem applies to the proposed clustering method.

\begin{theorem}\label{theo1}
    The proposed sequential clustering method is optimal to achieve the minimum equivalent noise power.
\end{theorem}

\begin{proof}
    According to (\ref{eq17}), we find that the equivalent noise power depends on the worst channel condition among the devices in each cluster.  The optimality of the method can be easily demonstrated using a proof by contradiction, which is omitted here due to the page limits.
    \hfill $\square$
\end{proof}

To increase the overall contribution of honest devices, an imperative is to congregate the Byzantine devices within the same clusters. This facilitates the participation of as many honest devices as possible in the aggregation procedure and thereby maximizing the sum contribution. In the following lemma, we demonstrate the superiority of the sequential clustering scheme in terms of enhancing the contribution of honest devices.

\begin{lemma}\label{pro1}
    With the optimal weighting factor setting, adopting the sequential clustering method maximizes the sum contribution, i.e., $\sum_{k\in\mathcal{S}_m} v_{m,k}$.
\end{lemma}

\begin{proof}
In the scenario where the PS processes detailed information about Byzantine and inactive devices, it is evident that setting the weighting factors of such devices to 0 is optimal for maximizing $\bar{g}_m$. This entails disregarding updates from these devices and amplifying the influence of the remaining honest devices. Simultaneously, setting the weighting factor to 0 causes the equivalent channel values of these devices to tend to infinity. Consequently, these devices are grouped together and the honest devices can participate effectively in training without being contaminated, thus maximizing the sum contributions $\sum_{k\in\mathcal{S}_m} v_{m,k}$.\hfill $\square$
\end{proof}

According to \emph{Theorem \ref{theo1}} and \emph{Lemma~\ref{pro1}}, we confirm the optimality of the proposed sequential clustering method for a given set of weighting factors. Substituting the optimal $\mathcal{C}_{m,n}\left(\{\alpha_{m,k}\}\right)$ into the original problem in (P2),  we transform the clustering optimization problem into a weighting factor optimization problem without compromising optimality, as discussed in the subsequent subsection.

%This transformation enables implicit control of clustering through weighting factor adjustment.
%Nevertheless, to ensure optimality, a crucial consideration is the identification of Byzantine devices, which .

\subsection{Weighting Factor Optimization}
Now, we are ready to address the weighting factor optimization that results from substituting in the optimal clustering:
\begin{align}
(\text{P4})\enspace \mathop{\text{maximize}}\limits_{\{\alpha_{m,k}\}_{k\in [K]}} \quad & V \bar{g}_{m}+\sum_{k=1}^K q_{m,k}\bar{\gamma}_{m,k}\alpha_{m,k} \nonumber \\
\text{subject to}\quad &\text{C}_1.
\end{align}
According to the proof in \emph{Lemma \ref{pro1}}, it is optimal to assign zero weighting factors to the truncated devices with poor channel conditions and the identified Byzantine devices. Let $\mathcal{D}_m$ denote the set of activated and honest devices at the $m$-th round, which is a deterministic set independent of the weighting factors. Therefore, $\forall k\notin \mathcal{D}_m$, we set $\alpha_{m,k}=0$ and we only need to optimize $\alpha_{m,k}$ for $k\in \mathcal{D}_m$. The relationship between $\mathcal{D}_m$ and the other sets, as well as the effect of setting zero weighting factors in the sequential clustering method, is shown in Fig. \ref{fig:cluster}. 
\begin{figure}[!t]
  \centering
  \centerline{\includegraphics[width=3.2in]{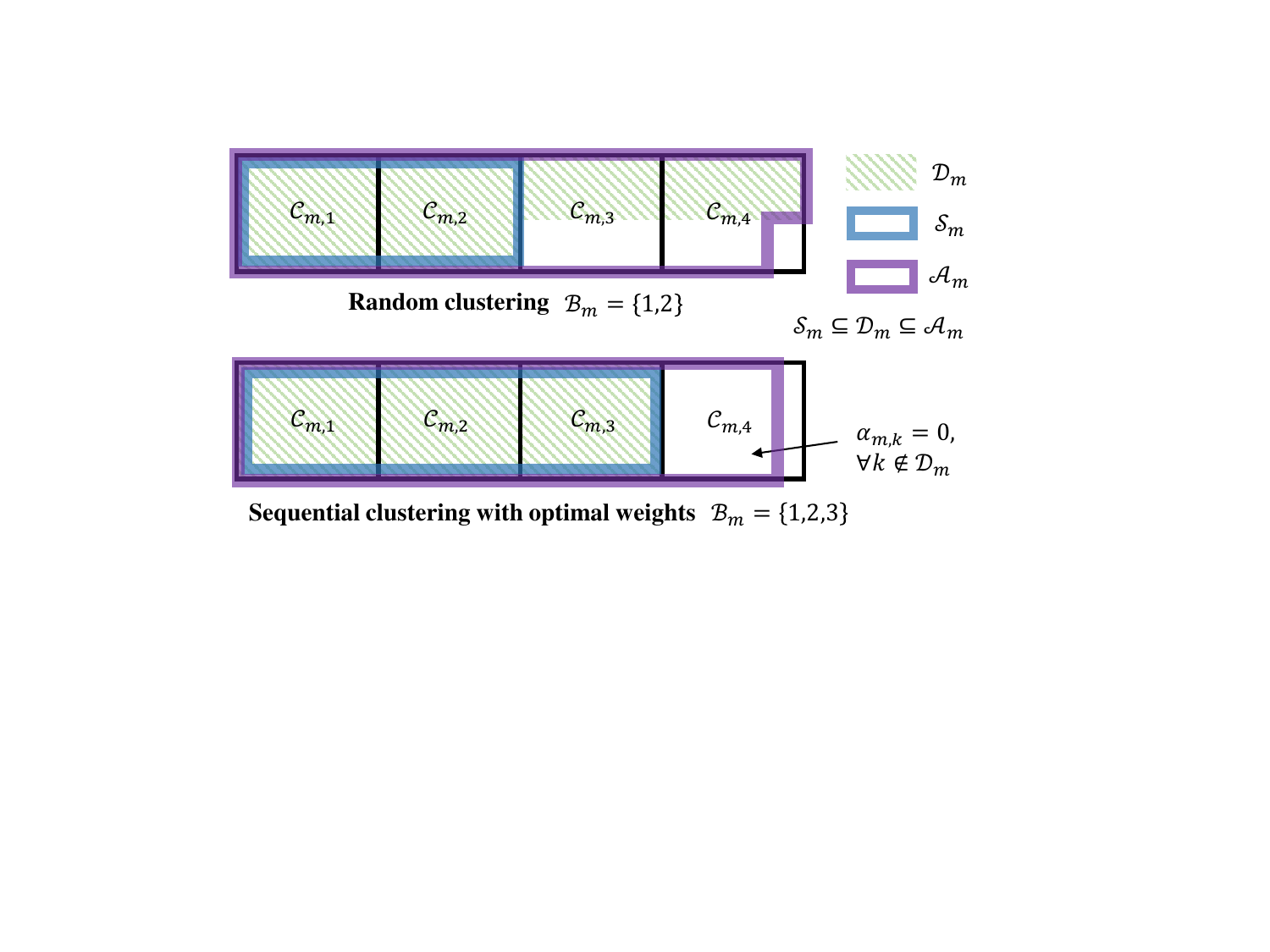}}
  \caption{Schematic diagram of the relationship between different device sets.}\label{fig:cluster}
\end{figure}
Moreover, the honest devices falling into Byzantine clusters will be contaminated and therefore do not require weight assignment. Hence, we have $\sum_{k\in \mathcal{S}_m}\alpha_{m,k}=\sum_{k\in \mathcal{D}_m} \alpha_{m,k}=1$. Then, by removing some constants and zero terms, the objective of problem (P4) can be further reformulated as in (\ref{eq23}) at the top of the next page,
\begin{figure*}
\begin{align} \label{eq23}
    f_m \triangleq \sum_{n\in \mathcal{B}_{m}\left(\{\alpha_{m,k}\}\right)}&\left( \sum_{k\in \mathcal{C}_{m,n}\left(\{\alpha_{m,k}\}\right)}\!\varphi_{m,k}\alpha_{m,k}-\max_{k\in \mathcal{C}_{m,n}\left(\{\alpha_{m,k}\}\right)} \!\left\{\varpi_{m,k}\alpha_{m,k}^2\right \}\right),
\end{align}
\hrulefill
\end{figure*}
where $\varphi_{m,k}=V\gamma_{m,k}+q_{m,k}\bar{\gamma}_{m,k}$ and $\varpi_{m,k}\!= \! \frac{VL  \eta \sigma^2G^2}{2(1\!-\!L\eta)P_{\max}\left\vert h_{m,k}\right\vert^2\beta_k^2}$. Note that both the specific clustering results $\mathcal{C}_{m,n}\left(\{\alpha_{m,k}\}\right)$ and the activated cluster set $\mathcal{B}_m\left(\{\alpha_{m,k}\}\right)$ depend on the weighting factors. Now, problem (P4) is equivalent to
\begin{align}
    (\text{P5})\enspace \mathop{\text{maximize}}\limits_{\{\alpha_{m,k}\}_{k\in \mathcal{D}_m}} \quad &  f_{m} \nonumber \\
    \text{subject to}\quad &\bar{\text{C}}_1: \sum_{k\in \mathcal{D}_m}\alpha_{m,k}=1.
\end{align}
Problem (P5) is neither convex nor smooth and thus challenging to solve.

%Regarding the norm value of gradients, $\Vert \mathbf{g}_m^k\Vert$, each device needs to send it to PS without error before performing AirComp \cite{nfzhang}.

We consider the ideal case in which the Byzantine devices in the current round have been accurately determined, and the clusters contaminated by them will be successfully identified and excluded from aggregation. Based on this assumption, we present the following theorem to facilitate the effective optimization of problem (P5).

\begin{theorem}\label{th2}
    Problem (P5) is equivalent to 
    \begin{align}
    (\text{P6}):&\mathop{\text{maximize}}\limits_{\substack{\{\alpha_{k}\}_{k\in \mathcal{D}_{m}}\\
    \{u_i\}_{i\in[\bar{N} ]}, \{l_i\}_{i\in[\bar{N}]}\\
    \{e_{i,k}\}_{i\in[\bar{N}],\, k\in \mathcal{D}_{m}}}}
    \enspace \sum_{i=1}^{\bar{N}} \sum_{k=1}^K \varphi_{k}\alpha_{k} e_{i,k}-\sum_{i=1}^{\bar{N}}  u_i^2 \nonumber \\
    &\text{subject to}\quad \bar{\text{C}}_1, \nonumber \\
    &\quad \quad \quad \quad \quad\text{C}_4: \sqrt{\varpi_k}\alpha_{k}e_{i,k}\leq u_i,\enspace \forall k,i,\nonumber \\
    &\quad \quad \quad \quad \quad\text{C}_5:\sqrt{\varpi_k}\alpha_{k}\geq l_i e_{i,k},\enspace \forall k,i,\nonumber \\
    &\quad \quad \quad \quad \quad\text{C}_6: l_{i-1}\geq u_i, \enspace i=2,\cdots,N,\nonumber \\
    &\quad \quad \quad \quad \quad\text{C}_7: e_{i,k}\in \{0,1\},\enspace \forall i,k,\nonumber \\
    &\quad \quad \quad \quad \quad\text{C}_8: \sum_{i=1}^{\bar{N}} e_{i,k}=1,\enspace \forall k,\nonumber \\
    &\quad \quad \quad \quad \quad\text{C}_9: \sum_{k \in \mathcal{D}_{m}} e_{i,k}=\bar{K},\enspace \forall i,
\end{align}
where $\alpha_{m,k}$, $\varphi_{m,k}$, and $\varpi_{m,k}$ are abbreviated as $\alpha_{k}$, $\varphi_{k}$, and $\varpi_{k}$, respectively, $\bar{N}\triangleq \left \vert \mathcal{B}_m \right \vert$, and $u_i$, $l_i$, and $e_{i,k}$ are auxiliary variables. In particular, $u_i$ and $l_i$ respectively represents the upper and lower bounds for the value of $\sqrt{\varpi_k\alpha_k}$ in cluster $i$. $e_{i,k}$ indicates whether the $k$-th device is in the $i$-th cluster.
\end{theorem}
\begin{proof}
    Please refer to Appendix \ref{app_th2}. \hfill $\square$
\end{proof}

Problem (P5) is a nonconvex MINLP problem of an intractable form. To tackle it with the discrete constraint $\text{C}_7$, we decompose $\text{C}_7$ into two equivalent constraints as follows:
\begin{align}\label{eq19}
    \text{C}_7&\Leftrightarrow e_{i,k}(e_{i,k}-1)=0,\enspace\forall i,k,\nonumber \\
    &\Leftrightarrow \left \{
    \begin{array}{ll}
         \text{C}_{10}: e_{i,k}(e_{i,k}-1)\geq 0,&  \forall i,k,\\
         \text{C}_{11}: e_{i,k}(e_{i,k}-1)\leq 0,&  \forall i,k.
    \end{array}
    \right.
\end{align}
By replacing constraint $\text{C}_7$ with constraints $\text{C}_{10}$ and $\text{C}_{11}$, problem (P3) is transformed into a continuous but nonconvex problem. To handle the nonconvexity, we adopt the P-CCP framework\cite{pccp}, which is proved to be effective in obtaining a feasible stationary point of the original problem in an iterative manner. Specifically, constraints $\text{C}_{10}$ and $\text{C}_{11}$ in (\ref{eq19}) are reformulated as
\begin{align}
    &\bar{\text{C}}_{10}:e_{i,k}(e_{i,k}-1)\leq  p_{i,k},\enspace \forall i,k, \nonumber \\
    &\bar{\text{C}}_{11}:(2e_{i,k}^{(n)}-1)e_{i,k}-\left(e_{i,k}^{(n)}\right)^2\geq - p_{i+N,k+K}, \enspace \forall i,k,\nonumber \\
    &\text{C}_{12}:\mathbf{p}\geq 0,
\end{align}
where $\mathbf{p}=[p_{1,1},\cdots,p_{2N,2K}]^T$ is a slack vector with
nonnegative elements, and $e_{i,k}^{(n)}$ is the solution obtained in the $n$-th iteration. Then, to guarantee the satisfaction of (\ref{eq19}), we impose penalty terms with respect to $\mathbf{p}$ in the original objective function, i.e., $\sum_{i=1}^{\bar{N}} \sum_{k=1}^K \varphi_{k}\alpha_{k} e_{i,k}-\sum_{i=1}^{\bar{N}}  u_i^2-\tau^{(n)} \Vert \mathbf{p} \Vert_1$, where $\tau^{(n)}$ is the penalty parameter in the $n$-th iteration. With a  sufficiently large penalty factor, the constraints in (\ref{eq19}) can be strictly met.

Moreover, the variable $e_{i,k}$ is coupled with $\alpha_k$ and $l_i$ in constraints $\text{C}_4$ and $\text{C}_5$, which further complicates the problem. To tackle the coupled variables, we reformulate the function $\alpha_k e_{i,k}$ as $\frac{1}{4}\left(\alpha_k +e_{i,k}\right )^2-\frac{1}{4}\left(\alpha_k -e_{i,k}\right )^2$. Although this decouples the variables, the function is nonconvex. Hence, to find an effective linear approximation, we introduce the following upper bound for the function:
\begin{align}
    &U^{(n)}(\alpha_k,e_{i,k})\!\triangleq \!\frac{1}{4}\left(\alpha_k\!+\!e_{i,k}\right )^2\!-\!\frac{1}{2}\left(\alpha_k^{(n)}\!-\!e_{i,k}^{(n)}\right) (\alpha_k\!-\! e_{i,k})\nonumber\\
    &\quad \quad \quad\quad \quad \quad \enspace +\frac{1}{4}\left(\alpha_k^{(n)}-e_{i,k}^{(n)}\right)^2,
\end{align}
where  $\alpha_k^{(n)}$ is the solution obtained for $\alpha_k$ in the $n$-th iteration. Then, the term $\alpha_k e_{i,k}$ in the objective of (P5) can be replaced by its upper bound $U^{(n)}(\alpha_k,e_{i,k})$ and constraint $\text{C}_4$ becomes
\begin{align}
    \bar{\text{C}}_4: \sqrt{\varpi_k} U^{(n)}(\alpha_k,e_{i,k}) \leq u_i,\enspace \forall k,i.
\end{align}
Similarly, constraint $\text{C}_5$ is replaced by
\begin{align}
    \bar{\text{C}}_5: \sqrt{\varpi_k} \alpha_k  \geq U^{(n)}(l_i,e_{i,k}) ,\enspace \forall k,i.
\end{align}

\begin{figure*}[!t]
  \centering
  \centerline{\includegraphics[width=6.2in]{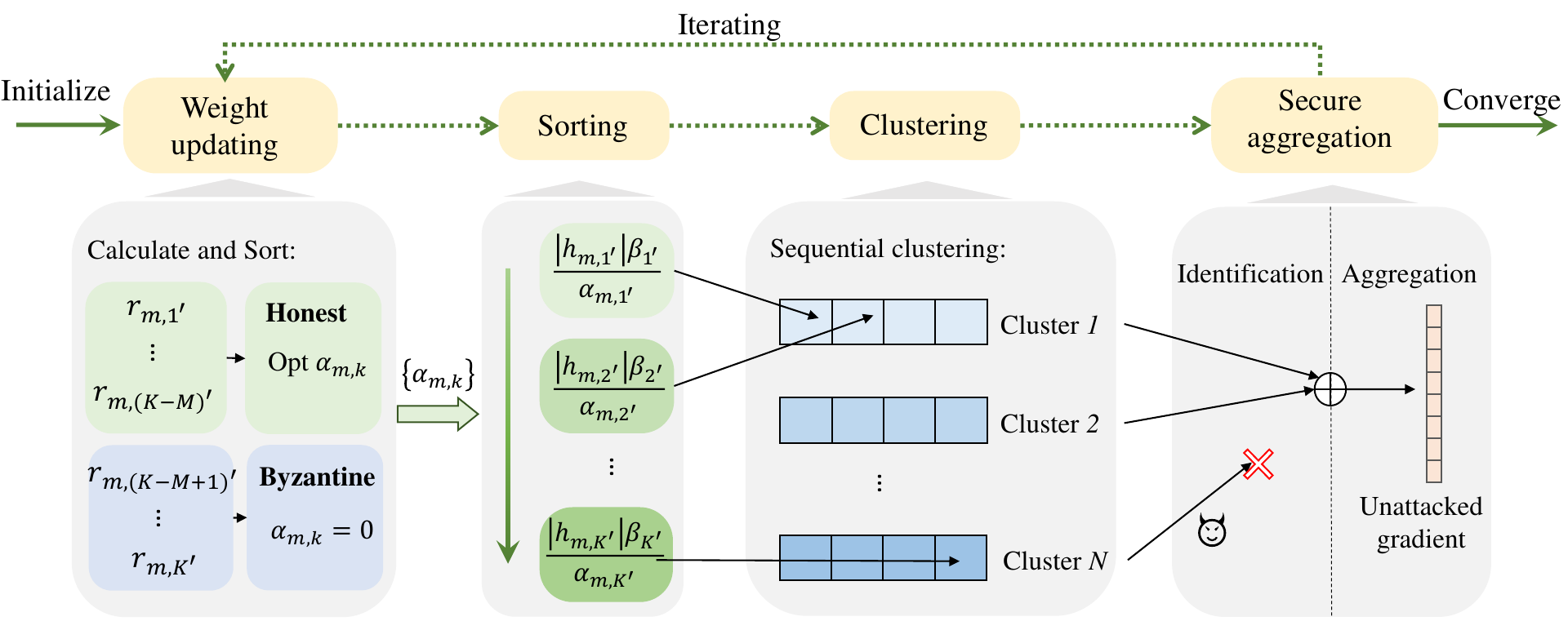}}
  \caption{Architecture of the proposed ZTA-empowered FedSAC approach.}\label{fig:res}
\end{figure*}

Applying these transformations, we construct a sequence of convex surrogate problems, the $n$-th of which is
\begin{align} \label{eq24}
    (\text{P7}):&\mathop{\text{maximize}}\limits_{\substack{\{\alpha_{k}\}_{k\in \mathcal{D}_{m}}\\
    \{u_i\}_{i\in[\bar{N}]}, \{l_i\}_{i\in[\bar{N}]}\\
    \{e_{i,k}\}_{i\in[\bar{N}],\, k\in\mathcal{D}_{m}},\mathbf{p}}}
    \enspace  \tilde{f}_{m}^{(n)}\nonumber \\
    &\text{subject to}\quad \bar{\text{C}}_1, \bar{\text{C}}_4, \bar{\text{C}}_5,\text{C}_6, \text{C}_8,\text{C}_9,\bar{\text{C}}_{10}, \bar{\text{C}}_{11},\text{C}_{12},
\end{align}
where $\tilde{f}_{m}^{(n)}\triangleq \sum_{i=1}^{\bar{N}} \sum_{k=1}^K \varphi_{k}L^{(n)}(\alpha_{k}, e_{i,k})-\sum_{i=1}^{\bar{N}}  u_i^2 -\tau^{(n)} \Vert \mathbf{p} \Vert_1$. Problem (P6) can be efficiently solved by numerical convex program solvers, e.g., CVX
\cite{grant2009cvx}. According to \cite{pccp}, the computational complexity of solving (P7) is $\mathcal{O}\left(K^{3.5}N^{3.5} \right)$.

To summarize, we conclude the steps in Algorithm \ref{alg0}. The penalty parameter is set as a small value at first to seek better performance, and then increases proportionally with the coefficient $\mu$ to avoid constraint violation. The penalty is upper-bounded by $\tau^{\max}$ is an upper limit of the penalty parameter, which helps the algorithm avoid numerical instability. The iteration terminates when the norm of the violation vector $\mathbf{p}$ is sufficiently small, i.e., $\Vert \mathbf{p}\Vert_1 \leq \chi$, so that the integer constraint is approximately satisfied.

\begin{algorithm}[!t]
\caption{P-CCP for Weighting Optimization} \label{alg0}
\begin{algorithmic}[1]  
\REQUIRE $n=0$, $\rho^{(0)}$, $\mu>1$, $e_{i,k}^{(0)}$, $\alpha_k^{(0)}$, and $l_i^{(0)}$.
\REPEAT
\STATE  Update $e_{i,k}^{(n+1)}$, $\alpha_k^{(n+1)}$, and $l_i^{(n+1)}$ by solving the problem in (\ref{eq24}).
\STATE Update $\tau^{(n+1)}=\min\left\{ \mu\tau^{(n)},\tau^{\mathrm{max}} \right \}$.
\STATE Update  $n=n+1$.
\UNTIL$\Vert \mathbf{p}\Vert_1 \leq \chi$ and the objective converges.
\end{algorithmic} 
\end{algorithm}

\subsection{Overall Framework and Performance Analysis}
The overall framework of the proposed FedSAC approach is depicted in Fig. \ref{fig:res}. At the beginning of each communication round, according to the previously updated reputation, the PS identifies the honest and Byzantine devices. For honest devices, their weighting factors are jointly optimized via Algorithm \ref{alg0} to maximize the loss reduction under the fairness constraints. With optimized weighting factors, the sequential clustering method is then applied. Finally, hierarchical AirComp with device clustering, combined with the clustering-based robust aggregation scheme, are conducted, thereby obtaining unattacked gradients for model updating.
In summary, we conclude the FedaSAC-based AirFL algorithm in Algorithm~\ref{alg1}.

%In each communication round, the PS employs Algorithm \ref{alg0} to optimize weighting factors, exploiting the gradient information reported by the devices and the identification of the Byzantine devices. Subsequently, based on the weights and channel conditions, the PS executes the sequential clustering. After the reception of aggregated gradients from all clusters, the robust aggregation algorithm is employed to derive effective gradients for updating. 

In the following theorem, we evaluate the convergence performance of AirFL using FedSAC by deriving the optimality gap, which is the difference between the optimal and actual achieved value of the loss function.

\begin{algorithm}[!t]
  \caption{FedSAC-based AirFL algorithm}
  \label{alg1}
  \begin{algorithmic}[1]
  \REQUIRE the initial model parameters, $\mathbf{w}_0$, the iteration number, $m=0$, the number of initialization rounds, $m_0$, the initial reputation $r_{0,k}=0$, and weighting factors $\alpha_{0,k}=\frac{1}{K}$ for all $k$.
  \REPEAT
  \STATE The PS refreshes the virtual queue in (19) and the reputation according to (21).
  \STATE The PS broadcasts $\mathbf{w}_m$ to all devices.
  \STATE Each device  computes the local gradient in (3) based on their local dataset and $\mathbf{w}_m$.
  \IF {$m\geq m_0$}
    \STATE Each device reports $\Vert \mathbf{g}_m^k \Vert$ to the PS.
    \STATE The weighting factors are optimized by Algorithm 1.
  \ENDIF
  \STATE The PS performs sequential clustering to group the devices into $N$ clusters according to the equivalent channel conditions.
  \FOR{$i=1,\cdots,N$}
  \STATE The devices in cluster $i$ perform AirComp according to their weighting factors.
  \STATE The PS receives the aggregated gradient of cluster $i$.
  \ENDFOR
  \STATE The PS discards the malicious gradients and conducts the model update according to (12).
  \STATE Increment $m\gets m+1$.
  \UNTIL{converges} 
  \end{algorithmic}
\end{algorithm}

\begin{theorem}\label{theo3}
    For a fixed learning rate satisfying $\eta\leq \min\left\{\frac{1}{L},\frac{1}{\mu}\right\}$, the optimality gap after $T$ communication rounds satisfies
    \begin{align}\label{eq32}
        \mathbb{E}\left[F(\mathbf{w}_{T})\right]-F^*  \leq& \prod_{m=0}^{T-1}B_m \left(F(\mathbf{w}_0) -F^*\right)\nonumber \\
        &+\sum_{i=0}^{T-2}C_i \prod_{j=i+1}^{T-1}B_j+C_{T-1},
    \end{align}
    where $B_m\triangleq 1-\mu \eta \sum_{k\in\mathcal{S}_m}\alpha_{m,k}$ and $C_m \triangleq \frac{\eta}{2}\sum_{k\in \mathcal{S}_m}\alpha_{k,m} \delta_k^2 -\frac{\eta\!-\!L\eta^2}{2}\sum_{k\in\mathcal{S}_m}\alpha_{m,k}G^2+\frac{L\eta^2}{2}\sigma_m^2$.
\end{theorem}

\begin{proof}
    Please refer to Appendix \ref{app_th3}. \hfill $\square$
\end{proof}

\begin{remark}
    According to (\ref{eq32}), the convergence rate of FL is solely contingent on parameter $B_m$, while the optimality gap is influenced by both $B_m$ and $C_m$. The proposed FedSAC approach allocates all weighting factors to activated honest devices through Byzantine identification, minimizing $B_m$ and thereby accelerating the convergence process. On the other hand, for $C_m$, through sequential clustering and weighting optimization, the total contribution is increased and the noise is reduced, resulting in an enhancement of the final achievable optimality gap. This result substantiates the effectiveness of the FedSAC approach.
\end{remark}

\section{Numerical Results}
\label{sec:typestyle}

In this section, we present numerical simulations to evaluate the performance of the proposed FedSAC scheme. The large-scale fading $\beta_{k}$ is modelled as $d_{k}^{-1.1}$, where $d_{k}$ denotes the distance between the PS and device $k$ and is uniformly distributed between 150 and 500 m. Moreover, we assume that the Byzantine devices transmit at full power to maximize resistance to convergence. Unless otherwise specified, the other parameters are set as follows: number of devices $K=40$, number of Byzantine devices $M=6$, noise power $\sigma^2 =10^{-6}$, maximum transmit power $P_{\max}=0$ dBm, penalty factor $V=10^5$, and the number of initialization rounds $m_0=10$.

For FL performance evaluation, we consider the image classification task based on the well-known MNIST and CIFAR-10 datasets. The MNIST dataset contains 10 classes of handwritten digits ranging from 0 to 9 and we train a  multi-layer perceptron (MLP) for it via AirFL. The MLP contains three fully-connected layers and 23,860 parameters. The CIFAR-10 dataset also includes 10 classes with labels 0-9 and we train a convolutional neural network (CNN) for it. The trained CNN is made up of two convolutional layers and three fully connected layers. We use max pooling operation following each convolutional layer and the RELU activation function. Without loss of generality, we adopt relatively simple models for simulation. Since our proposed FedSAC approach is confined to the physical layer and remains independent of specific models or algorithms, it can also be applied to more complicated scenarios. Moreover, it is noteworthy that all local datasets are non-IID, exclusively comprising a single category of labels.  We set $\delta_k$ via simulation tests. The learning rate is set as $\eta=0.005$. We consider the following three typical types of Byzantine attacks \cite{tspbay,onebit}:
\begin{itemize}
    \item Sign-flipping attack: the Byzantine devices set $\tilde{\mathbf{g}}_{m,k}=-\sum_{k\in \mathcal{H}}\mathbf{g}_{m,k}$, where $\mathcal{H}$ denotes the set of honest devices, attempting to guide the trained model towards a direction opposite the desired one.
    \item Gaussian attack: the Byzantine devices randomly select $\tilde{\mathbf{g}}_{m,k}$ from a Gaussian distribution with mean $\mathbf{1}$, where $\mathbf{1}$ denotes an  all-one vector.
    \item Label-flipping attack: the Byzantine devices reverse the labels of their local datasets, by replacing the label $x$ of each training sample with $9-x$, and setting $\tilde{\mathbf{g}}_{m,k}$ as the gradient calculated by the incorrect labels.
\end{itemize}

\begin{figure*}[!t]
	\centering
	\begin{minipage}[t]{0.32\linewidth}
		\centering
		\includegraphics[width=1\linewidth]{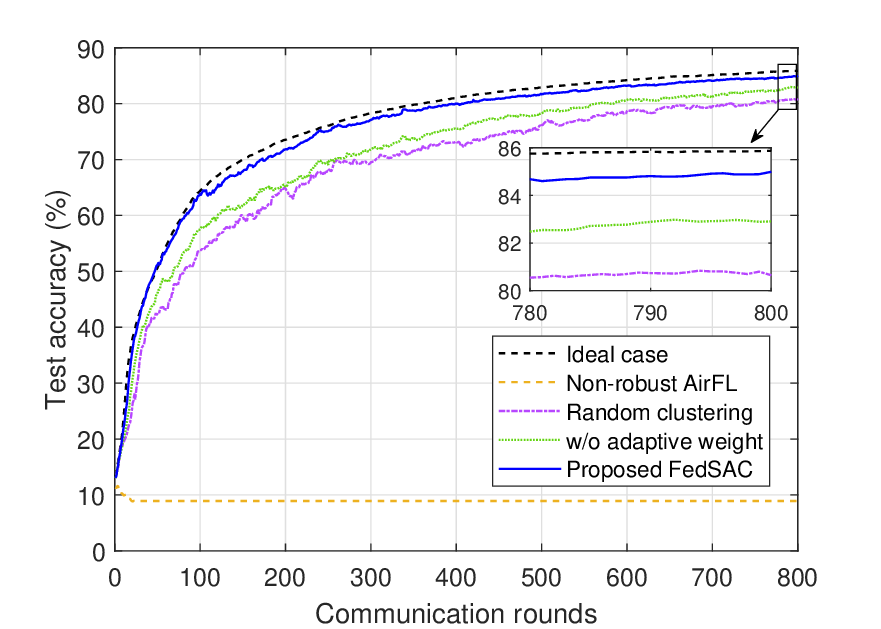}
	\end{minipage}
	\begin{minipage}[t]{0.32\linewidth}
		\centering
		\includegraphics[width=1\linewidth]{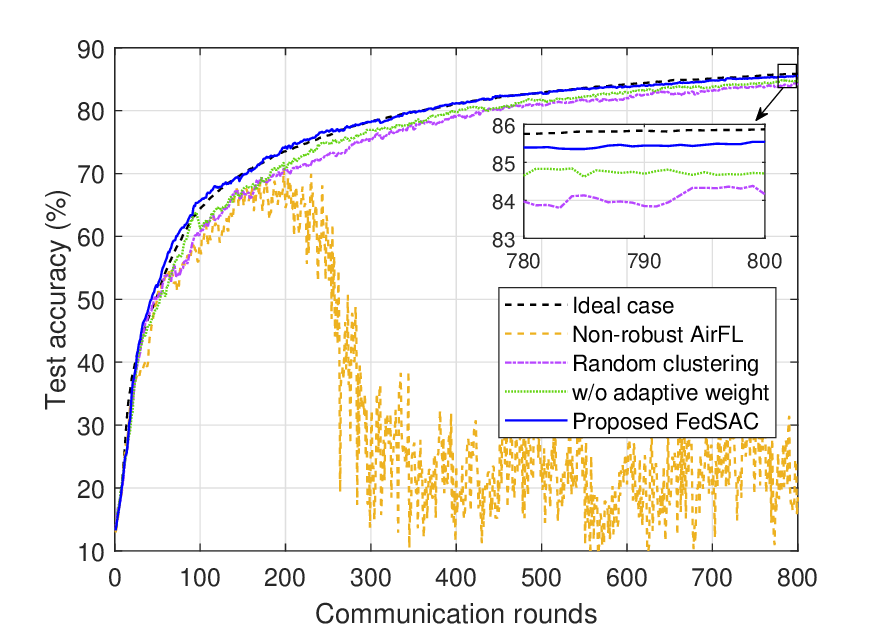}
	\end{minipage}
	\begin{minipage}[t]{0.32\linewidth}
		\centering
		\includegraphics[width=1\linewidth]{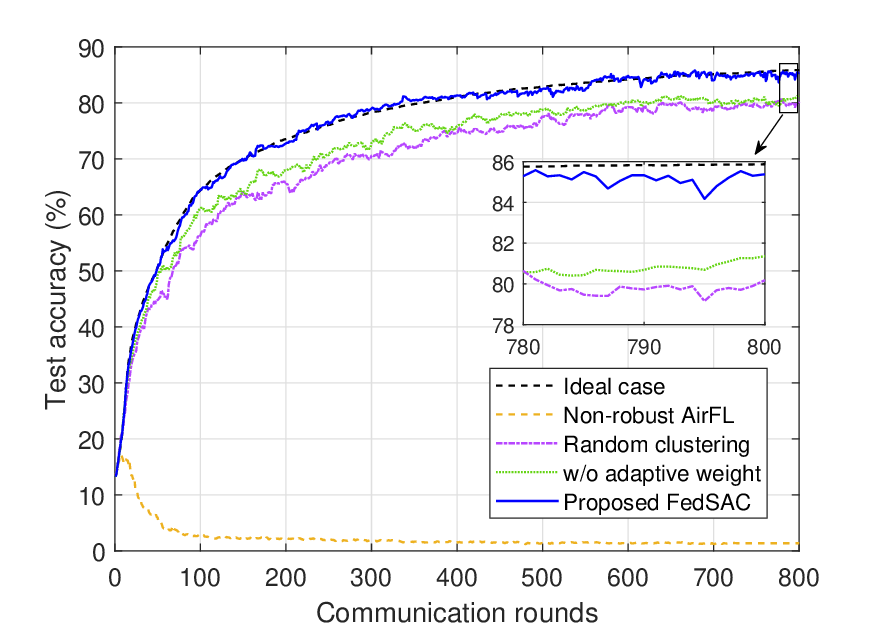}
	\end{minipage}
 \caption{Test accuracy for the MNIST dataset and different types of attacks. (a) Sign-flipping attack (b) Gaussian attack (c) Label-flipping attack\label{fig3}.}
\end{figure*}

\begin{figure*}[!t]
	\centering
	\begin{minipage}[t]{0.32\linewidth}
		\centering
		\includegraphics[width=1\linewidth]{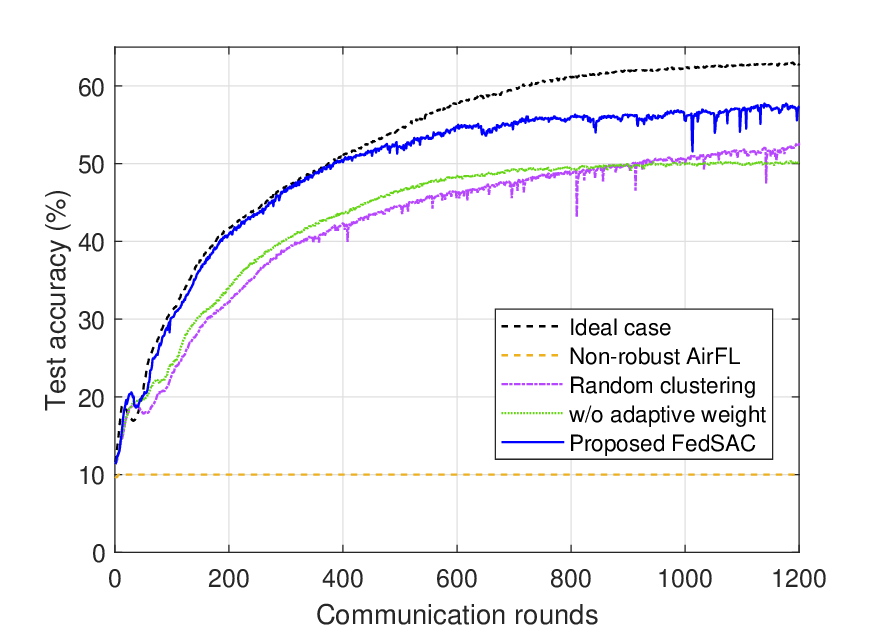}
	\end{minipage}
	\begin{minipage}[t]{0.32\linewidth}
		\centering
		\includegraphics[width=1\linewidth]{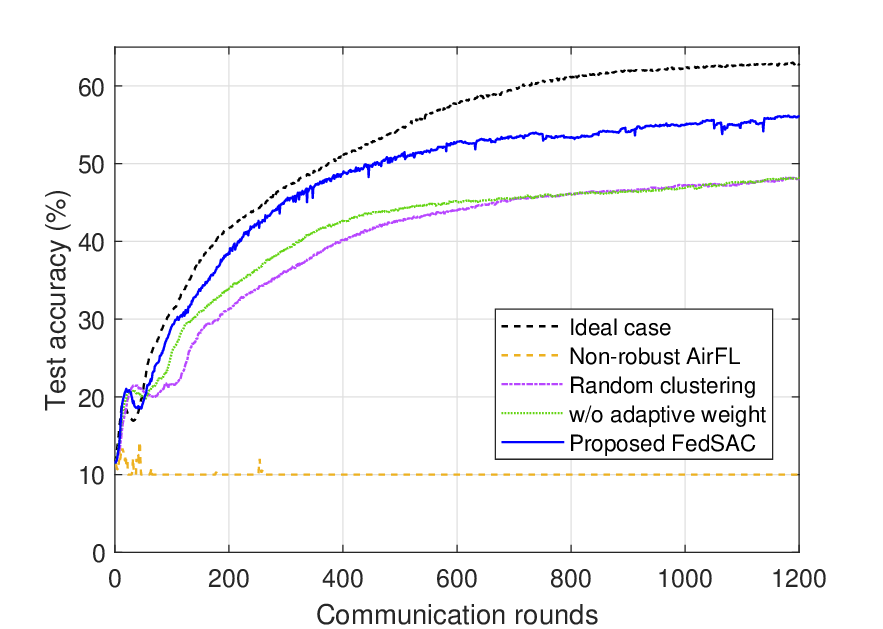}
	\end{minipage}
	\begin{minipage}[t]{0.32\linewidth}
		\centering
		\includegraphics[width=1\linewidth]{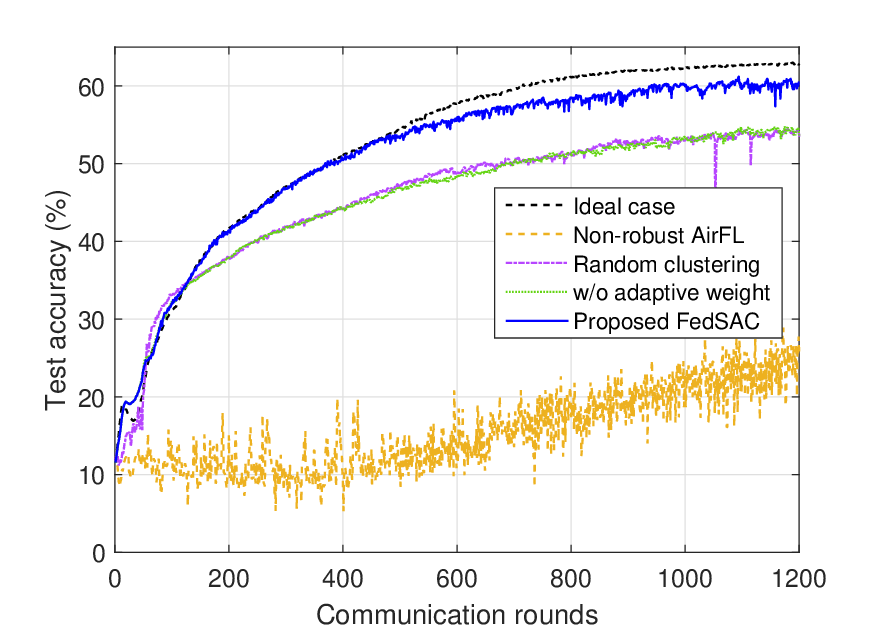}
	\end{minipage}
 \caption{Test accuracy for the CIFAR-10 dataset and different types of attacks. (a) Sign-flipping attack, (b) Gaussian attack, and (c) Label-flipping attack\label{fig4}.}
\end{figure*}

We compare our proposed algorithms with the following baseline schemes:
\begin{itemize}
    \item Ideal case: There are no Byzantine devices nor additive noise, representing the ideal FL setting. Its performance serves as an upper bound for the proposed approach.
    \item Non-robust AirFL\cite{imperfect}: No additional robust design is performed and the traditional AirFL scheme is used. It can be understood as the lower bound and highlights the necessity of robust design.
    \item Random clustering \cite{park,liye1,liye2}: Clustering-based robust aggregation is performed and device clustering is randomly determined.
    \item w/o adaptive weighting: The sequential clustering is conducted only for noise minimization and the weighting factors are set as $\frac{1}{K}$.
\end{itemize}
The last two baselines can be interpreted as ablation experiments of the proposed approach, thus validating the effectiveness of each component in the proposed FedSAC. It is worth emphasizing that the proposed approach does not limit the robust aggregation method, and the main innovation focuses on the design of device clustering. Therefore, without loss of generality, we consider a representative robust aggregation method and the main focus of the numerical comparisons is on the performance gains brought by different device clustering methods.
Moreover, the implementation details and source codes are available online at \cite{sourcecode}.

\subsection{Learning Performance for MNIST Dataset}
We first evaluate the learning performance with the MNIST dataset in Fig. \ref{fig3}. FedSAC is seen to outperform the existing baselines across all types of attacks, with a negligible performance gap compared to the ideal case. For the case of the sign-flipping attack, FedSAC exhibits a notable enhancement of approximately 4.5\% in terms of test accuracy when juxtaposed with existing random clustering schemes. FedSAC achieves performance comparable to random clustering in approximately 420 communication rounds, thereby reducing the number of rounds by 47.5\%. In the other two scenarios, the performance gain provided by FedSAC is smaller, although it still achieves a notable improvement in accuracy and accelerated convergence. This discrepancy arises due to the fact that the employed similarity detection method can more easily detect the sign-flipping attack, which enhances the accuracy and efficiency of Byzantine identification. In the case of Gaussian and label-flipping attacks, the PS requires more communication rounds to identify the Byzantine devices, leading to a reduction in the observed gain. Certainly, the observed performance loss could potentially be compensated by enhancing the robust aggregation algorithm adopted. However, such improvements extend beyond the primary scope of our research and are consequently not included in this study. In the initial phases of training, it is noteworthy that FedSAC can exhibit performance that surpasses even the ideal case. This is attributed to the adjustment of the weighting factors based on device contribution, promoting faster convergence than that of the design with equal weights.

It is also evident that a performance gain can be achieved solely through the sequential clustering approach, due to the reduction in equivalent noise. Through a comparative analysis of the three methods, the ablative experiment validates the roles of each module in FedSAC.

\subsection{Learning Performance for CIFAR-10 Dataset}
Fig. \ref{fig4} depicts the test accuracy of different schemes for the CIFAR-10 dataset. 
As in the previous case, FedSAC maintains exceptional performance under diverse attack scenarios, although the performance gap compared with the ideal case is widened. This is due to the increased complexity of the classification task for the CIFAR-10 dataset compared to MNIST and and the heightened difficulty of learning an accurate global model for heterogeneous local datasets.
In the initial training phases, dynamic weight optimization leads the PS to prioritize devices with greater contributions, accelerating convergence and yielding performance comparable to the ideal situation. Subsequently, the imbalance in the weights may introduce bias towards specific devices, particularly for non-IID local datasets, causing a significant performance gap. Owing to the relative simplicity of the task for the MNIST dataset, the impact of this bias is not significant, resulting in no pronounced performance degradation.

In contrast to the baseline schemes with attackers, the heightened importance of weight optimization is evident in its implicit regulation of clustering. This leads to fewer honest gradients being contaminated, yielding substantial performance gains. In the presence of the three types of attacks, FedSAC  demonstrates improvements of 6.2\%, 7.1\%, and 5.9\% in teasting accuracy compared with the random clustering method. Moreover, FedSAC reduces the number of communication rounds needed to achieve equivalent accuracy by 60\%, 67.5\%, and 54\%, further affirming its outstanding performance. It is worth noting that in the CIFAR-10 dataset, relying solely on sequential clustering does not yield significant performance gains and may even result in degradation. This is because by prioritizing devices based on their channel conditions, the sequential clustering may cause an imbalance in the participation of different devices, thus causing a biased global model. To elaborate, honest devices sharing similar channels with Byzantine devices are more likely to be placed in the same cluster as the Byzantine devices. This leads to their reduced participation in the training process, thus elucidating the need for incorporating fairness in the  optimization.

\subsection{Impact of Transmit Power}
    In Fig. \ref{fig5}, we compare the test accuracy for different transmit power budgets using the MNIST dataset under sign-flipping attack. We see that as the transmit power decreases, the performance degradation of FedSAC is comparatively less pronounced, showcasing its robustness to noise. In contrast, especially with random clustering schemes, the test accuracy experiences a significant decline at lower transmit power levels. This result is due to the  sequential clustering method which reduces the equivalent noise power and improves the utilization of transmit power, thereby providing a significant advantage at low SNRs.  In high SNR regions where the influence of noise is minimal, the performance gain primarily stems from weighting allocation optimization based on device contributions.

\subsection{Impact of the Number of Byzantine Devices}
Fig. \ref{fig6} exhibits the test accuracy versus the number of Byzantine devices for the different schemes. 
As the number of Byzantine devices increases, the performance of methods degrades. In contrast to the baseline schemes, FedSAC showcases superior robustness against the Byzantine attack. Despite the escalating number of Byzantine devices, the proposed approach is able to effectively group them into specific clusters, resulting only in a marginal performance degradation. Conversely, both random grouping and sequential clustering experience a steep decline in performance as the number of attacking devices increases.

\begin{figure}[!t]
\centering
\includegraphics[width=3.2in]{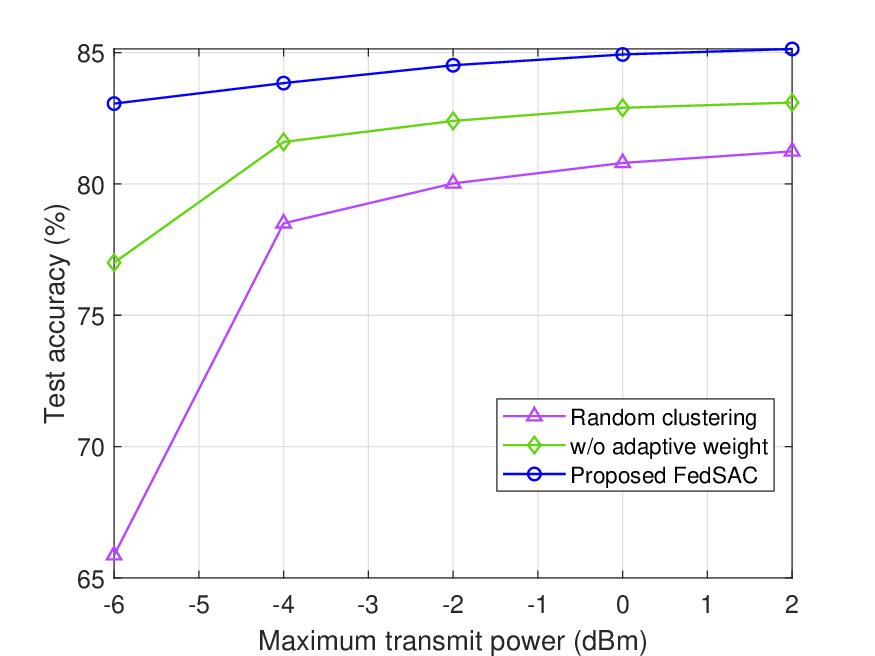}
\caption{ Test accuracy versus maximum transmit power\label{fig5}.}
\end{figure}

\begin{figure}[!t]
\centering
\includegraphics[width=3.2in]{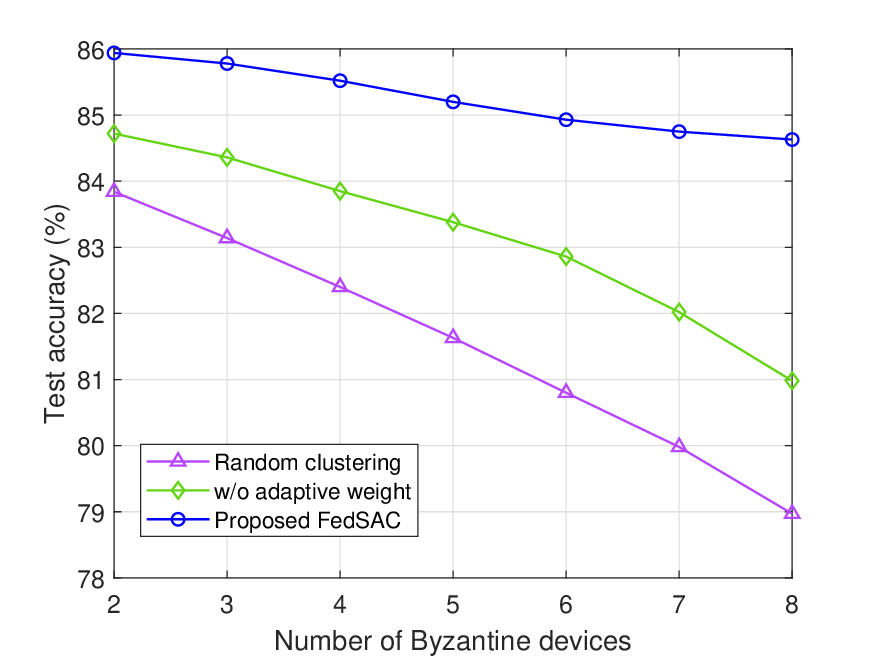}
\caption{ Test accuracy versus the number of Byzantine devices.}\label{fig6}
\end{figure}

\subsection{Impact of the Number of Clusters}
In Fig. \ref{fig7}, we investigate the influence of the number of clusters. For fairness considerations, we presume equal spectrum resource allocation across all clusters, with distinct time slots designated for AirComp within each cluster. Subsequently, we evaluate the test accuracy after an equivalent training duration, concentrating solely on the AirComp-induced delay and disregarding local computation delays. As the number of clusters is increased, there is enhanced support for a more efficient Byzantine-robust aggregation algorithm, leading to a reduction in the number of honest gradients contaminated in each round and consequently improving performance. Nevertheless,  this escalation concomitantly entails a heightened consumption of time-frequency resources, diminishing the number of communication rounds within the specified latency constraints, thereby introducing adverse effects. This is evident that the test accuracy depicted in Fig. \ref{fig7} exhibits an initial ascent followed by a decline as the number of clusters increases. 
The different methods exhibit maximum test accuracy at distinct cluster counts. Notably, FedSAC attains optimality with two clusters, whereas random clustering  requires five clusters for optimal performance. This discrepancy arises from the incorporation of Byzantine identification and adaptive clustering in FedSAC, enabling efficient concentration of the Byzantine devices within the same cluster. Consequently, even with only two clusters, effective participation of partial honest devices is ensured. In contrast, random clustering results in the even distribution of Byzantine devices across the clusters, motivating the need for more clusters to enhance reliability. From this standpoint, FedSAC demonstrates more efficient resource utilization and exhibits a notable performance advantage compared to the baseline schemes.

\begin{figure}[!t]
\centering
\includegraphics[width=3.2in]{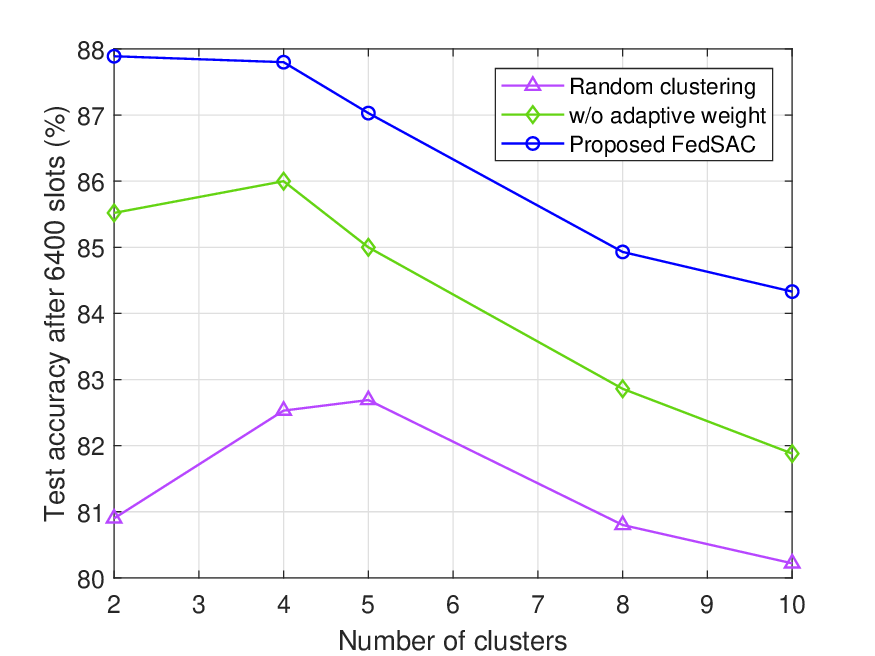}
\caption{ Test accuracy versus the number of clusters.}\label{fig7}
\end{figure}

\begin{figure}[!t]
\centering
\includegraphics[width=3.2in]{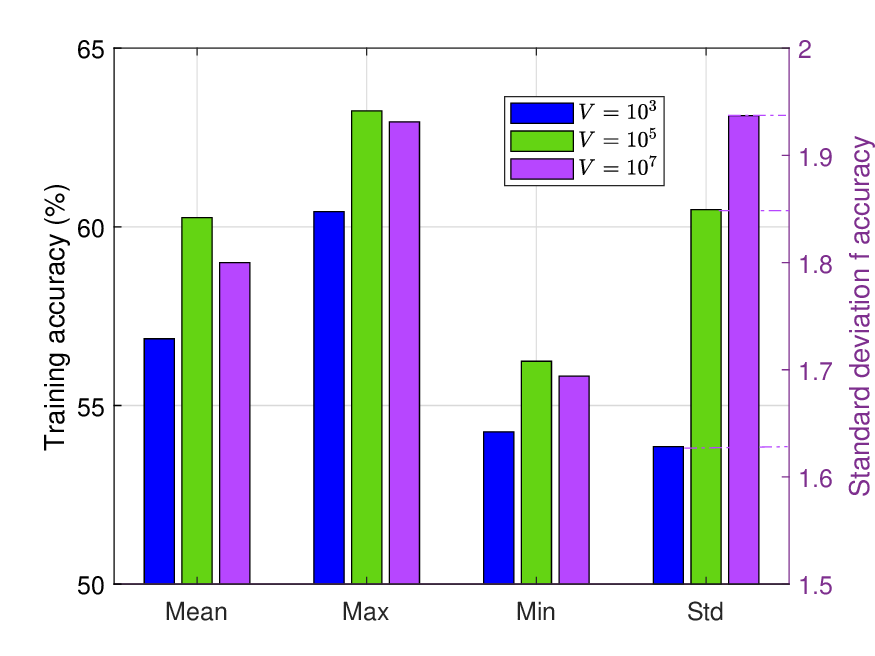}
\caption{ Statistics of different devices' training accuracy for different $V$\label{fig8}.}
\end{figure}

\subsection{Impact of Penalty Factor $V$}

By varying the penalty factor $V$ in the Lyapunov transformation, we effectively tune for fairness among various devices. In Fig. \ref{fig8}, we illustrate the effect of the penalty factors on the training accuracy of distinct devices.  Given the increased sensitivity to heterogeneity in the classification task for the CIFAR-10 dataset, and recognizing its consequential impact on fairness, we employ the CIFAR-10 dataset for evaluation. We use the standard deviation of the training accuracy among the involved honest devices as the measure of fairness. Evidently, with an increase in $V$, the training accuracy gap among devices becomes more conspicuous, indicating a deterioration in fairness. A excessive value for $V$ may lead to a reduction in fairness, resulting in a biased global model and diminished performance. Conversely, an overly small $V$ can give rise to a loss in performance by placing undue emphasis on fairness while overlooking the influence of channel and gradient characteristics in the current round. Hence, it is imperative to judiciously select the value of $V$ for better performance.

\section{Conclusion}
\label{sec:conclusion}
In this paper, we proposed a novel ZTA-based FedSAC scheme for AirFL to combat Byzantine attacks. Based on clustering, we constructed multiple individual gradients from different clusters and facilitated the implementation of existing robust aggregation methods. Based on the one-step convergence analysis, we formulated a joint optimization problem involving clustering and weighting factor allocation with theoretical guarantees. The incorporation of Byzantine device identification, sequential clustering, and adaptive weighting effectively addresses the difficulties associated with  the optimization problem and plays a pivotal role in achieving substantial performance improvements. The simulation results verified the effectiveness of the proposed FedSAC approach and confirmed its superiority over baseline approaches from various perspectives. The proposed approach offers a practical solution for deploying AirFL with ZTA in resource-constrained and vulnerable wireless networks, highlighting the crucial role of Byzantine attack defense in real-world applications.

\appendices 
\section{Proof of Lemma \ref{lemma1}}\label{app_lemma1}
To begin with, by exploiting \emph{Assumption 1}, we bound the expected per round loss $g_m$ as
\begin{align}\label{app1}
    g_m&\leq \mathbb{E}\left[ \nabla^T F(\mathbf{w}_m) (\mathbf{w}_{m+1}\!-\!\mathbf{w}_m)\!+\!\frac{L}{2}\left \Vert \mathbf{w}_{m+1}\!-\!\mathbf{w}_m\right \Vert^2 \right]\nonumber \\
    &\overset{\text{(a)}}{=}\underbrace{-\eta \nabla^T F(\mathbf{w}_m) \sum_{k\in \mathcal{S}_m}\alpha_{m,k}\mathbf{g}_m^k}_{A_1} \nonumber \\
    &\quad+\underbrace{\frac{L\eta^2}{2}\left \Vert \sum_{k\in \mathcal{S}_m}\!\alpha_{m,k}\mathbf{g}_{m,k}\right \Vert^2 }_{A_2}+\frac{L\eta^2}{2}\mathbb{E}\left[\left \Vert \sum_{n\in\mathcal{B}_m} \mathbf{z}_n\right \Vert^2\right]\nonumber\\
    &=A_1+A_2+\frac{L\eta^2}{2}\sigma_m^2,
\end{align}
where (a) comes from (\ref{eq11}). Next, we first rewrite $A_1$ as
\begin{align}\label{app2}
    A_1&=-\eta \sum_{k\in\mathcal{S}_m}\alpha_{m,k}\nabla^T F(\mathbf{w}_m) \mathbf{g}_m^k\nonumber \\
    &=-\frac{\eta}{2} \sum_{k\in\mathcal{S}_m}\alpha_{m,k} \left(\left \Vert \nabla F(\mathbf{w}_m) \right \Vert^2+\left \Vert \mathbf{g}_m^k \right \Vert^2 \right.\nonumber \\
    &\left.\quad-\left \Vert \nabla F(\mathbf{w}_m)-\mathbf{g}_m^k \right \Vert^2 \right) \nonumber \\
    &\overset{\text{(b)}}{\leq} -\frac{\eta}{2} \sum_{k\in\mathcal{S}_m}\alpha_{m,k} \left(\left \Vert \nabla F(\mathbf{w}_m) \right \Vert^2+\left \Vert \mathbf{g}_m^k \right \Vert^2 \!-\!\delta_k^2 \right),
\end{align}
where (b) exploits \emph{Assumption 2}. As for $A_2$, we have 
\begin{align}\label{app3}
A_2&\overset{\text{(c)}}{\leq } \frac{L\eta^2}{2}\sum_{k\in\mathcal{S}_m}\alpha_{m,k} \left \Vert \mathbf{g}_m^k \right \Vert^2,
\end{align}
where (c) is due to the convexity of the function $\Vert \cdot \Vert^2$ and the fact that $\sum_{k\in\mathcal{S}_m}\alpha_{m,k}\leq 1$. Plugging (\ref{app2}) and (\ref{app3}) into (\ref{app1}), we complete the proof. 

\section{Proof of Theorem \ref{th2}} \label{app_th2}
First, under the assumption of perfect identification,  we can calculate the number of clusters that can effectively participate in the aggregation, i.e., $\bar{N} =  N-\frac{M}{\bar{K}}$,
where $\frac{M}{\bar{K}}$ is assumed here to be an integer for simplicity. Note that the objective function in (\ref{eq23}) is related to the specific clustering results, which are implicitly determined by the weighting factors. For a more tractable form, we introduce the new variables $ \{e_{i,k}\}_{i\in[\bar{N}],\, k\in \mathcal{D}_{m}}$ to characterize the clustering results. Specifically, if $e_{i,k}$ is equal to $1$, it signifies that the $k$-th device is located in cluster $i$, and otherwise it is equal to $0$. Also, considering the max-min form of the noise-related term, we introduce $u_i$ and $l_i$ to respectively represent the maximum and minimum value of $\sqrt{\varpi_{k}}\alpha_k$ in cluster $n$. With the above variables, the objective in (\ref{eq23}) is reformulated as
\begin{align}
    \sum_{i=1}^{\bar{N}} \sum_{k=1}^K \varphi_{k}\alpha_{k} e_{i,k}-\sum_{i=1}^{\bar{N}}  u_i^2.
\end{align}

To ensure equivalence with the original formulation, we need to add the following constraints for the auxiliary variables introduced above. To begin with, since each device is located in only one cluster and each cluster contains $\bar{K}$ devices, the variables $ \{e_{i,k}\}_{i\in[\bar{N}],\, k\in \mathcal{D}_{m}}$ must satisfy $\sum_{i=1}^{\bar{N}}e_{i,k}=1$, and $\sum_{k\in\mathcal{D}_{m}}e_{i,k}=\bar{K}$. Furthermore, the variable $u_i$ is defined as an upper bound for $\sqrt{\varpi_{k}}\alpha_{k}$ in cluster $n$. Hence, it is constrained by
\begin{align}
    \sqrt{\varpi_{k}}\alpha_{k} e_{i,k}\leq u_i,\enspace \forall i,
\end{align}
where $e_{i,k}$ is used to used to determine whether device $k$ is within cluster $i$ and the constraint always holds if $e_{i,k}=0$. Likewise, the lower bound $l_i$ must satisfy $\sqrt{\varpi_{k}}\alpha_{k} \geq l_i e_{i,k},\enspace \forall k$. Furthermore, considering that the specific clustering results are determined by the proposed sequential clustering strategy, there is a sequential relationship between the equivalent channels of the devices within different clusters. Specifically, the equivalent channel of a device within a previous cluster is always worse than the one in the following cluster, i.e., with a larger $\sqrt{\varpi_{k}}\alpha_k$. Thus, to satisfy the clustering rule, the lower bound for $\sqrt{\varpi_{k}}\alpha_k$ in the previous cluster needs to be greater than the maximum value of the following cluster, namely $l_{i-1}\geq u_i$. Combining all transformations and constraints, we arrive at problem (P3) and the proof completes.

\section{Proof of Theorem \ref{theo3}} \label{app_th3}
Substituting \emph{Assumption 3} into (\ref{eq14}) and exploiting the fact that $\Vert \mathbf{g}_{m}^k\Vert \leq G$, we have
\begin{align}
    g_m &\leq -\mu \eta \left (\sum_{k\in \mathcal{S}_m}\alpha_{m,k}\right )\left(F(\mathbf{w}_{m})-F^*\right)+\frac{\eta}{2}\sum_{k\in \mathcal{S}_m}\alpha_{m,k} \delta_k^2 \nonumber \\
    & \quad-\frac{\eta\!-\!L\eta^2}{2}\sum_{k\in\mathcal{S}_m}\alpha_{m,k}\left \Vert  \mathbf{g}_m^k \right \Vert^2+\frac{L\eta^2}{2}\sigma_m^2\nonumber \\
    &\leq -\mu \eta  \left (\sum_{k\in\mathcal{S}_m}\alpha_{m,k} \right)\left(F(\mathbf{w}_{m})-F^*\right)+C_m,
\end{align}
where $C_m \triangleq \frac{\eta}{2}\sum_{k\in \mathcal{S}_m}\alpha_{k,m} \delta_k^2 -\frac{\eta\!-\!L\eta^2}{2}\sum_{k\in\mathcal{S}_m}\alpha_{m,k}G^2+\frac{L\eta^2}{2}\sigma_m^2$.
Letting $B_m=1-\mu \eta \sum_{k\in\mathcal{S}_m}\alpha_{m,k}$, we further obtain
\begin{align}
    \mathbb{E}&\left[F(\mathbf{w}_{T})\right]-F^* \nonumber \\
    &\leq B_{T-1} \left(\mathbb{E}\left[F(\mathbf{w}_{T-1})\right]-F^*\right)+C_{T-1}\nonumber \\
    &\leq B_{T-1}\left(B_{T-2} \left(\mathbb{E}\left[F(\mathbf{w}_{T-1})\right]-F^*\right)+C_{T-2}\right)+C_{T-1}\nonumber \\
    & \leq \!\prod_{m=0}^{T-1}\!B_m \!\left(F(\mathbf{w}_0)\! -\!F^*\right)+\sum_{i=0}^{T-2}C_i \prod_{j=i+1}^{T-1}B_j+C_{T-1}.
\end{align}
This completes the proof.

\bibliographystyle{IEEEtran}
\bibliography{strings}

\end{document}